\documentclass[acmsmall,nonacm]{acmart}

\usepackage[linesnumbered,noend]{algorithm2e}
\usepackage{xspace}
\usepackage{subcaption} 
\usepackage{listings}
\graphicspath{{Figures/}}

\SetKwProg{Fn}{Function}{}{}
\SetKwProg{Fn}{function}{}{end}
\SetKwProg{Struct}{struct}{}{end}
\SetArgSty{}
\SetProgSty{}
\SetFuncArgSty{}
\SetKwInOut{Input}{input}
\SetKwInOut{Output}{output}
\SetKw{And}{and}

\def\kernel{\textsf{Ker}\xspace}
\def\genpaths{\textsf{genpaths}\xspace}

\def\subspaceclosure{\textsf{subspace\_closure}\xspace}

\def\pathtoQ{\textsf{sub\_paramQ\_bypartition}\xspace}
\def\wavefronttoQ{\textsf{sub\_paramQ\_bywavefront}\xspace}
\def\maxcover{\textsf{combine\_subQ}\xspace}
\def\Ginterf{\mathcal{Q}\xspace}
\def\mayspill{\textsf{may-spill}\xspace}
\def\addparam{\textsf{combine\_paramQ}\xspace}
\def\mainloop{\textsf{program\_Q}\xspace}

\def\dfg{DFG\xspace}

\def\tool{\textsc{IOLB}\xspace}
\def\polybench{\textsc{PolyBench}\xspace}

\usepackage{tikz}
\usetikzlibrary{snakes,arrows,shapes,positioning,calc,tikzmark}
\usepackage{xcolor}

\title{Automated Derivation of Parametric Data Movement Lower Bounds for Affine Programs}

\begin{document}

\author{Auguste~Olivry}
\affiliation{
	\institution{Univ. Grenoble Alpes, Inria, CNRS, Grenoble INP, LIG, 38000 Grenoble, France}
}
\author{Julien~Langou}
\affiliation{
	\institution{University of Denver Colorado}
}
\author{Louis-No{\"e}l~Pouchet}
\affiliation{
	\institution{Colorado State University}
}

\author{P.~Sadayappan}
\affiliation{
	\institution{University of Utah}
}
\author{Fabrice~Rastello}
\affiliation{
	\institution{Univ. Grenoble Alpes, Inria, CNRS, Grenoble INP, LIG, 38000 Grenoble, France}
}

\authorsaddresses{}

\def\S{S} 
\def\T{T} 
\def\ST{\left(\S+\T\right)}
\def\K{K} 
\def\P{{\mathcal P}\xspace}
\def\q{\textit{q}\xspace}
\def\U{\textit{U}\xspace} 
\def\LB{\textit{L}\xspace} 
\def\IIO{\textit{IIO}\xspace} 
\def\EIO{\textit{EIO}\xspace}
\newcommand*{\EE}{\mathcal{E}}
\def\IO{\textit{I/O}\xspace}
\def\ID{\textit{D\hspace{-0.5em}I\hspace{0.25em}}\xspace}
\def\comp{\textsf{comp}}
\def\frontier{\textsf{input}}
\def\interf{\textsf{interf}}
\def\coeffInterf{\textsf{coeffInterf}}
\def\isbroadcast{\textsf{isBroadcast}}
\def\iscircuit{\textsf{isChain}}
\def\OI{\textit{OI}\xspace}
\def\CI{\textit{CI}\xspace}
\def\OIup{\textit{OI}_{\textrm{up}}}
\def\OIopt{\textit{OI}_{\textrm{manual}}}
\def\OIpluto{\textit{OI}_{\textrm{PLuTo}}}
\def\MB{\textit{MB}\xspace}
\def\peakB{\textit{BW}_\textit{mc}}
\def\peakF{\textit{GFLOPS}_\textit{mc}}
\def\i{\textit{I}\xspace}
\def\I{\mathcal{I}\xspace}
\def\Q{Q}
\def\Qlow{Q_{\textrm{low}}}
\def\Qinfty{Q_{\textrm{low}}^\infty}
\def\Qpluto{Q_{\textrm{PluTo}}}
\def\Qopt{Q_{\textrm{manual}}}
\def\R{\mathbb{R}}
\def\Z{\mathbb{Z}}
\def\load{\textsc{Load}\xspace}
\def\loads{\textsc{Loads}\xspace}
\def\stores{\textsc{Stores}\xspace}
\def\Ker{\textrm{Ker}}
\def\Rwf{R_{\textrm{wf}}}
\newcommand{\In}[1]{\textrm{In}\left(#1\right)}
\newcommand{\card}[1]{\left|#1\right|}
\newcommand{\weight}[2]{\textrm{weight}_{#1}\left(#2\right)}
\newcommand{\proj}[2]{\textrm{proj}_{#1}\left(#2\right)}
\newcommand{\msources}[1]{\mathrm{Sources}\left(#1\right)}
\def\Sources{Sources\xspace}
\def\sources{sources\xspace}
\renewcommand{\dim}[1]{\mathrm{dim}\left(#1\right)}
\def\inset{In-set\xspace}
\def\DFG{Data-flow graph\xspace}
\def\dfg{DFG\xspace}
\def\Lineage{Lineage\xspace}
\def\lineage{lineage\xspace}
\def\lineages{lineages\xspace}
\def\Lineages{Lineages\xspace}
\def\alive{live\xspace}
\def\Alive{Live\xspace}
\def\broken{broken\xspace}
\def\V{\mathcal{S}}
\def\sched{\mathcal{O}}
\def\E{\mathcal{D}}
\newcommand{\yes}{{\color{green}\ding{51}}}
\newcommand{\no}{{\color{red}\ding{55}}}
\newcommand{\soso}{{\color{orange}\ding{51}}}
\newcommand{\closin}[1]{\overline{#1}}
\def\IR{\mathcal{C}}
\def\QQ{\mathcal{Q}}
\def\QI{\mathcal{C}}
\newcommand{\domain}[1]{\mathrm{Dom}\left(#1\right)}
\newcommand{\image}[1]{\mathrm{Im}\left(#1\right)}
\def\B{\mathcal{B}}
\def\L{\mathcal{L}}
\newcommand{\rk}[1]{\mathrm{rank}\left(#1\right)}
\def\latphi{\mathcal{L}_{\phi_1,\phi_2, \dots, \phi_m}}
\newcommand{\ddx}[2]{\frac{\partial #1}{\partial #2}}
\newcommand{\gengrp}[1]{\left\langle #1 \right\rangle}
\def\RR{\mathcal{R}}
\def\RL{R\hspace{-0.2em}L}

\begin{abstract}
  
  For most relevant computation, the energy and time needed for data movement dominates that for performing arithmetic operations on all computing systems today. 
  Hence it is of critical importance to understand the minimal total data movement achievable during the execution of an algorithm. 
  The achieved total data movement for different schedules of an algorithm can vary widely depending on how efficiently the cache is used, e.g., untiled versus effectively tiled matrix-matrix multiplication.
  A significant current challenge is that no existing tool is able to meaningfully quantify 
  the potential reduction to the data movement of a computation
  that can be achieved by more effective use of the cache through operation rescheduling.
  Asymptotic parametric expressions of data movement lower bounds have previously been manually derived for a limited number of algorithms, often without scaling constants.
  In this paper, we present the first compile-time approach for deriving non-asymptotic 
  parametric expressions of data movement lower bounds for arbitrary affine computations.
  
  The approach has been implemented in a fully automatic tool (\tool) that can generate these lower bounds for input affine programs. 
  \tool's use is demonstrated by exercising it on all the benchmarks of the \polybench suite. The advantages of \tool are many: (1) \tool enables us to derive bounds for few dozens of algorithms for which these lower bounds have never been derived. This reflects an increase of productivity by automation. (2) Anyone is able to obtain these lower bounds through \tool, no expertise is required. (3) For some of the most well-studied algorithms, the lower bounds obtained by \tool are higher than any previously reported manually derived lower bounds.
  
\end{abstract}

\maketitle

\section{Introduction}
The cost of performing arithmetic/logic operations on current processors is significantly lower than the cost of moving data from memory to the ALU (arithmetic/logic units), whether measured in terms of latency, throughput, or energy expended. 
While the impact of operations and data movement latency can be effectively masked by issuing a sufficient number of independent operations for pipelined hardware functional units and memory, the maximum throughput of ALUs and memory imposes fundamental limits to achievable performance. Similarly, the minimum volume of data movement imposes fundamental limits to the energy required for a computation.

  The \emph{operational intensity} (\OI), defined as the ratio of the number of arithmetic operations to the volume of data movement to/from memory, is a critical metric for the data movement created by the schedule of an algorithm. 
  
  A convenient way to characterize an algorithm's performance bottlenecks on a particular processor is to compare the \OI of the program to 
  the \emph{machine balance parameter} (\MB), relating peak computational rate to peak data transfer rate to/from memory.
  If a program's achieved \OI is less than \MB, it will be memory-bandwidth limited and therefore constrained to much lower performance than the machine peak.
  Since \MB has been steadily increasing over the years, this progressively changes programs from being compute-bound to becoming memory-bandwidth limited.

The volume of data transfer for a program is not a fixed quantity and could potentially be improved via program transformations that change the schedule of operations in order to efficiently use the \emph{cache}. 
For example, it can be easily shown that for matrix multiplication of two $N \times N$ matrices, any of the six possible loop permutations of the untiled code will necessarily incur movement of at least $N^3$ \emph{words} (data elements) from memory if $N^2 > \S$, the cache capacity. 
This means that the operational intensity for untiled matrix multiplication cannot be higher than $\frac{2N^3}{N^3}=2$ FLOPs per word. 
In contrast, a tiled execution can lower the volume of data moved to $2N^3/\sqrt{\S}$ words, increasing the \OI from $2$ to $\sqrt{\S}$ FLOPs per word.

Performance tools like Intel's Software Development Emulator Toolkit (SDE) and VTune Amplifier (VTune) enable the measurement of the achieved \OI of a program. 
The measured \OI is then typically used to assess the quality of the current implementation. 
For example, if the measured \OI is much lower than \MB, then we expect the performance of the code to be greatly constrained by memory bandwidth. 
However, as illustrated by the example of matrix-multiplication, the achieved \OI of an algorithm can vary significantly across different functionally equivalent implementations.
An important question is: 
{\bf How can we know whether the achieved \OI for a particular implementation of an algorithm is close to the maximum possible for that algorithm or whether significant improvement is potentially feasible?} 
No existing performance tools can help answer this question. 
In this paper, we take a significant step towards addressing this fundamental question, by developing a novel compile-time approach to automatically establish upper bounds on achievable \OI for any input affine program.

In substance, the tool (called \tool) automatically derives parametric lower bounds (with {\em scaling constants}) on the data transfer volume (and thus also provides a parametric upper bound on achievable \OI) for any scheduling of an arbitrary affine code on a two-level memory system.
\tool can be viewed as a proof environment, where the input is an affine code, and the output is an \IO lower bound for this affine code for any valid schedule of operations. 
The formal proof itself can be derived, understood, and reviewed from the output of the \tool algorithm. 
The lower bound is expressed as a function of the parameters of the affine code.
A lower bound is optimal/tight when we can find a (valid) schedule of operations that realizes it.

The lower bound found by \tool is not necessarily tight, but it is always valid.
We have applied \tool on the \polybench test suite~\cite{polybench} made of 30 affine codes.
Our implementation returns non trivial lower bounds including new, never published ones. 
This has been done automatically by the press of a button in less than a second on a basic computer.
In a few cases, the results it provides are better (higher) than previously published lower bounds. 
To assess its quality, we also derived communication-efficient schedules for the \polybench codes.
In 11 cases, the lower bound matches the one realized by the manually optimized schedule, assessing the optimality of the corresponding lower bounds and schedules.

The paper is organized as follows. 
A high-level overview of the approach is presented in Sec.~\ref{sec:highlevel}. 
The formalism for data movement lower bounds based on the seminal red-blue pebble game of Hong \& Kung~\cite{hong.81.stoc}, along with the core definitions and theorems used to derive our algorithm, are described in Sec.~\ref{sec:foundations}. 
Sec.~\ref{sec:decomp} provides insights on how complex programs can be decomposed to derive tighter bounds.
An overview of the complete framework is provided in Sec.~\ref{sec:complete}.
It uses two proof techniques, namely the $K$-partition and the wavefront based proofs that are respectively described in Sec.~\ref{sec:partition} and Sec.~\ref{sec:wf}.
We demonstrate the power of our approach by running it on a full benchmark suite of affine programs: Sec.~\ref{sec:exp} reports the data movement complexities for all benchmarks of the \polybench suite, compared to the data movement cost achieved by an optimizing compiler.
Related work and conclusions can be found in sections~\ref{sec:related} and~\ref{sec:conclusion}.

\section{Overview of Approach}
\label{sec:highlevel}

The developed compile-time analysis tool incorporates two very distinct approaches to finding lower bounds on data movement: 
a first one based on the so-called S-Partitioning approach \cite{hong.81.stoc} and a second one based on graph wavefronts \cite{elango-spaa2014}. 
Before delving into details on the compile-time analysis for automated derivation of data-movement lower bounds for arbitrary affine computations, we present a high-level overview of the first approach, the S-Partitioning approach, in order to familiarize the reader with the goal and terminology of the methods.

\begin{figure}
\captionsetup[subfigure]{justification=centering}
	\begin{minipage}{0.5\textwidth}
  \begin{subfigure}[b]{\textwidth}
    \centering
    \lstset{language=C,basicstyle=\ttfamily,escapeinside={<@}{@>}}
    {\footnotesize\begin{lstlisting}
<@\textbf{Parameters}@>: N, M; 
<@\textbf{Input}@>: A[N], C[M]; <@\textbf{Output}@>: A[N];
for(t=0;t<M;t++)
  for (i=0; i<N; i++) 
    A[i] = A[i] * C[t];
    \end{lstlisting}}
    \caption{C-like code\label{fig:ex-2d-code}}
  \end{subfigure}
  
  \begin{subfigure}[b]{\textwidth}
    \centering
    \lstset{language=C,basicstyle=\ttfamily,escapeinside={<@}{@>}}
    {\footnotesize\begin{lstlisting}
<@\textbf{Parameters}@>: N, M; 
<@\textbf{Input}@>: A[N], C[M]; <@\textbf{Output}@>: <@$S_{M-1}$@>[N];
for (<@$0\le t<M$@> and <@$0\le i<N$@>)
  if (t==0): <@$S_{0,i}$@> = A[i] * C[0];
  else: <@$S_{t,i}$@> = <@$S_{t-1,i}$@> * C[t];
    \end{lstlisting}}
		\caption{Corresponding single assignment form\label{fig:ex-2d-sa}}
	\end{subfigure}
	\end{minipage}
  \begin{subfigure}[c]{0.4\textwidth}
    \includegraphics[width=0.8\textwidth]{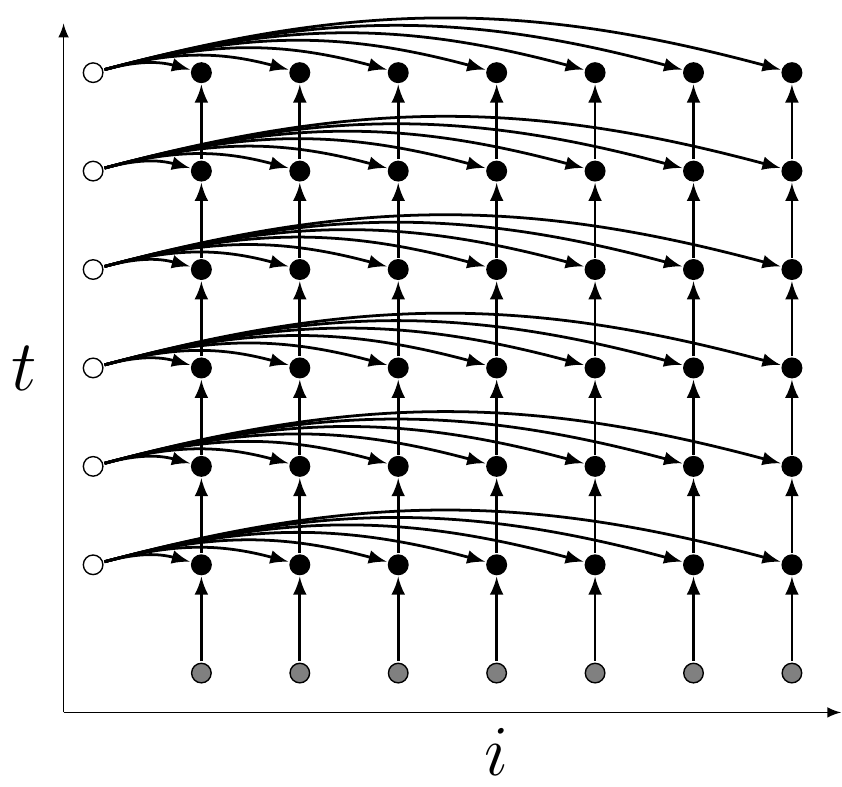}
    \caption{Corresponding CDAG. Input nodes \texttt{A[N]} (resp.~\texttt{C[N]}) are shown in grey (resp.~white) while compute node are shown in black\label{fig:ex-2d-cdag}}
  \end{subfigure}
	\caption{\label{fig:ex-2d}Example~1}
\end{figure}

\paragraph{Elementary example}
Consider the simple program in Fig.~\ref{fig:ex-2d-code}. 
For given values of the parameters M and N (e.g., M=6, N=7), the program can be abstracted as an explicit computational directed acyclic graph (CDAG), as shown in Fig.~\ref{fig:ex-2d-cdag}. 
Vertices in the CDAG represent input values for the computation as well as values computed by all statement instances (the latter are colored black and the former have lighter shades, grey or white). 
Edges in the CDAG capture data flow dependences, i.e., relations between producers of data values to consumers. 
We note that in this abstracted representation of the computation, there is no association of any memory locations with values. 
Fig.~\ref{fig:ex-2d-sa} shows a single-assignment form of the same computation as that in Fig.~\ref{fig:ex-2d-code}, and both programs have the same CDAG shown on Fig.~\ref{fig:ex-2d-cdag}. 
The CDAG abstracts all possible valid schedules of execution of the statement instances: 
the only requirement is that all predecessor vertices in the CDAG must be executed before a given vertex can be executed. 
Data movement is modeled in a simplified two-level memory hierarchy, with an explicitly controlled fast memory of limited size $\S$ (e.g., a set of registers or a scratch-pad), and a slow memory of unlimited capacity. 
At any point in the execution at most $\S$ values corresponding to CDAG vertices may be in fast memory. 
A computational CDAG vertex can be executed only if the values corresponding to all predecessor vertices are present in fast memory.

Consider any valid schedule for the execution of the vertices of a CDAG, expressed as a sequence of instructions: 
load, store, or operation execution (Op). 
A valid schedule must ensure that values corresponding to predecessor vertices are available in fast-memory when the operation corresponding to each CDAG vertex is executed. 
The sequence of instructions of the schedule is partitioned into contiguous maximal sub-sequences such that the total number of load instructions in any sub-sequence (except the last one) is exactly equal to a specified limit $\T$. 
Let us suppose (as explained shortly) that no more than $U$ Ops can be provably present within any of the partitioned sub-sequences. 
Let $V$ denote all computational vertices in the CDAG. 
There must be $\lfloor|V|/U\rfloor$ sub-sequences with $\T$ loads, leading to a lower bound on the number of loads of $Q^{\textrm{low}}=\T\cdot\lfloor|V|/U\rfloor$.

We next use the simple example of Fig.~\ref{fig:ex-2d-cdag} to explain how an upper bound for $U$ can be computed. 
The automated analysis based on partitioning in \tool is centered around the use of geometric inequalities that relate the cardinality of a set of points in a multi-dimensional space to cardinalities of lower-dimensional projections of those points. 
The set of points here ($P$) are the computational vertices (Ops) in one of the partitioned sub-sequences ($SS$) with $\T$ load instructions. 
The \inset $\In{P}$ of $P$ is the set of all predecessors of the vertices in $P$ that do not belong to $P$. 
Clearly, $\In{P}$ represents values that were not computed in the current sub-sequence $SS$ containing $P$. 
Since all values in $\In{P}$ must be in fast memory in order to execute the Ops corresponding to $P$, they must either have already been in fast memory at the beginning of the sub-sequence $SS$ or must have been explicitly loaded within $SS$. 
No more than $S$ values from $\In{P}$ could have been present at the beginning of $SS$, and $\T$ values were loaded in $SS$. 
Thus the size of $\In{P}$ must be less than $\ST$.

In our simple example, vertices corresponding to the loop statement are naturally represented as points in a two-dimensional lattice. 
With that representation, it may be observed that the size of the \inset of a vertex set must be greater than or equal to the cardinality of the orthogonal projections of $P$ onto the vertical and horizontal axes. 
As illustrated on Fig.~\ref{fig:loomis2d}, the size of the vertex set in the two-dimensional space is bounded by the product of the sizes of its two 1D projections. 
This result can be generalized to arbitrary dimensions and any set of projections, and is called the Brascamp-Lieb inequality. 
Setting $\T = \S$, a vertex set with an \inset of size at most $2\S$ cannot have projections of size more than $2\S$, and therefore cannot itself be greater than $U=4S^2$. 
This implies that any valid ordering of the operations for this computation will result in at least $S \cdot \lfloor MN/4S^2\rfloor \approx MN/4S$ load operations\footnote{It is actually possible to improve this bound by a factor of 4 with more advanced techniques, as shown in Sec.~\ref{sec:partition}}.

\paragraph{General approach} To automate and generalize this geometric reasoning on arbitrary affine programs, we need to:
1.~Generalize the geometric upper-bounding for any number of projections with arbitrary dimensionality (Thm.~\ref{thm:bl});
2.~Build (derive from array accesses) a compact representation (DFG) of the data-flow dependencies of the program that is suitable for reasoning about reuse directions (Sec.~\ref{sec:DFG});
3.~Analyze this representation to extract reuse directions (represented as DFG-paths -- Sec.~\ref{ssec:paths}, Alg.~\ref{alg:genpaths});
4.~Generalize the geometric reasoning for a perfectly nested loop with one statement to any combination of loops with arbitrary number of statements (embedding -- Def.~\ref{def:proj}).

The goal of \tool is to go even further and automatically derive parametric bounds that are as tight as possible (including maximization of the scaling constants).
For this purpose, the developed algorithm:
5.~Enables the combination (and tightening) of constraints associated with different projections, even with an arbitrary number of them with lower dimensionality (Lemma~\ref{lemma:optim}, Sec.~\ref{ssec:lb});
6.~Handles non-orthogonal projections even if they are not linearly independent (Lemma~\ref{lm:lattice}, Alg.~\ref{algo:lattice});
7.~Develops a new reasoning strategy inspired from the wavefront reasoning of Elango et al.~\cite{elango-spaa2014} (Sec.~\ref{sec:wf});
8.~Allows the combination of individual complexities of overlapping program regions (Def.~\ref{def:may-spill}, Lemma~\ref{lemma:composition}) even for an unbounded number of regions (parameterized regions inside loops -- Sec.~\ref{ssec:loop});

\section{Foundations}
\label{sec:foundations}

In this section, we present some background and discuss prior results needed for the developments in this paper.

\subsection{CDAG and \IO complexity}

The formalism and methodology we use to derive schedule-independent data movement lower bounds
for execution of an algorithm on a processor with a two-level memory hierarchy 
is strongly inspired by the foundational work of Hong \& Kung \cite{hong.81.stoc}.
In this formalism, an algorithm is abstracted by a graph --- called a CDAG ---, where vertices model
execution instances of arithmetic operations and edges model
data dependencies among the operations. 
The data movement (or \IO) complexity of a CDAG is formalized via the red-white
pebble game (a variation of Hong \& Kung's red-blue pebble game).
In this game, a vertex of a CDAG can hold red and white pebbles.
Red pebbles represent values in the fast memory (typically a cache or scratchpad),
and their total number is limited.
White pebbles represent computed values, that can be loaded into the fast memory.
A value can be computed only when all its operands reside in the fast memory: 
a red pebble can be placed on a vertex in the CDAG if all its predecessors hold a red pebble, a white pebble is placed alongside the red.
Values that have been computed can be loaded in and discarded from the fast memory at any time:
a red pebble can be placed or removed from a vertex holding a white pebble.
However a value can only be computed once: once a vertex holds a white pebble, it cannot be removed.
The \IO cost of an execution of the game is the number of loads into the fast memory: the number of times a red pebble is placed alongside a white one.

Contrary to Hong \& Kung's original model, our formalism \emph{does not allow recomputation} of the value at a vertex. 
This follows many previous efforts~\cite{BDHS11,BDHS11a,bilardi2001characterization,bilardi2012lower,DemmelGHL12,elango-spaa2014,elango-popl2015,toledo.jpdc,savage.cc.95}. 
This assumption is necessary to be able to derive bounds for complex CDAGs by decomposing them into subregions. 
Another slight difference of \tool with prior work is that it only models \emph{loads and not stores} --- this means the generated bounds are clearly also valid lower bounds for a model that counts both loads and stores. 
Since the number of loads dominates stores for most computations, the tightness of the lower bounds is not significantly affected.
We provide formal definitions below.

\begin{definition}[Computational Directed Acyclic Graph (CDAG)]
	A \emph{Computational Directed Acyclic Graph (CDAG)} is a tuple $ G = (V, E, I)$ of finite sets such that 
	$(V, E)$ is a directed acyclic graph, $I \subseteq V$ is called the \emph{input set}
	and every $v \in I$ has no incoming edges.
\end{definition}

\begin{definition}[Red-White Pebble Game]
	Given a CDAG $G = (V, E, I)$, we define a complete $\S$-red-white pebble game ($\S$-RW game for short)
	as follows:
	In the initial state, there is a white pebble on every input vertex $v \in I$, $S$ red pebbles
	and an unlimited number of white pebbles.
	Starting from this state, a complete game is a sequence of steps using the following rules, resulting
	in a final state with white pebbles on every vertex.

	\begin{description}
		\item [(R1)] A red pebble may be placed on any vertex that has a white pebble.
		\item [(R2)] If a vertex $v$ does not have a white pebble and all its immediate predecessors have red pebbles on them,
			a red pebble may be placed on $v$. A white pebble is placed alongside the red pebble.
		\item [(R3)] A red pebble may be removed from any vertex.
	\end{description}

	The \emph{cost} of a $\S$-RW game is the number of applications of rule (R1), corresponding to the number of transfers from slow to fast memory. 
\end{definition}

\begin{definition}[\IO complexity]
	The \emph{\IO} (or  \emph{data movement}) \emph{complexity} of a CDAG $G$ for a fast memory capacity $\S$, denoted $\Q(G)$, is the minimum cost of a complete $\S$-RW game on $G$.
\end{definition}

\subsection{Partitioning}\label{ssec:partitioning}
One key idea from Hong \& Kung was the design of a mapping between
any valid sequence of moves in the red-blue pebble game
and a convex partitioning of the vertices of a CDAG and 
thereby the assertion of an \IO lower bound for any valid schedule
in terms of the minimum possible count of the disjoint vertex-sets in any valid 
$2\S$-partition (see below) of the CDAG.

The argument is the following: any execution can be decomposed into consecutive segments doing exactly (but for the last one) $\S$ loads.
There are at most $\S$ vertices in fast memory before the start of each segment.
Considering the set of computed vertices in one of these segments, we can bound the size of its ``frontier'' by $2\S$:
there can be at most $\S$ vertices in fast memory before the execution of the segment,
and by construction there are exactly $\S$ loads.

Smith et al.~\cite{smith-gemm-19} introduced a generalization of this argument,
leading to tighter bounds in many cases.
The idea is to decompose the execution into segments with $\T$ loads.
This leads to a $\ST$-partitioning lemma instead of the original $2\S$.
We provide formal definitions below.

\begin{definition}[\inset]
	Let $G=(V,E)$ be a DAG, $P\subseteq V$ be a vertex set in $G$.
	The \inset of $P$ is the set of vertices outside $P$ with a successor inside $P$.
	Formally,
		\[\In{P}=\{v\in V\smallsetminus P,\ \exists (v,w)\in E \wedge w\in P\}\]
\end{definition}

\begin{definition}[$\K$-bounded set]
	Let $G = (V, E)$ be a DAG. A vertex set $P \subseteq V$ is called \emph{$\K$-bounded} if $\In{P} \leq \K$.
\end{definition}

\begin{definition}[$\K$-partition of a CDAG]
	Let $G = (V, E, I)$ be a CDAG.
	A $\K$-partition of $G$ is a collection of subsets $V_1, V_2, \dots, V_m$ of $V\setminus I$ such that:
	\begin{description}
		\item [(1)] $\{V_1, V_2, \dots, V_m\}$ is a partition of $V \setminus I$,
			i.e. $\forall i \ne j \ V_i \cap V_j = \emptyset$ and $\bigcup_{i=1}^m V_i = V \setminus I$.
		\item [(2)] There is no cyclic dependence between $V_i$'s.
		\item [(3)] Every $V_i$ is $\K$-bounded, i.e. $\forall i \ \In{V_i} \leq \K$.
	\end{description}
\end{definition}

\begin{lemma}[$\ST$-Partitioning \cite{hong.81.stoc}]\label{lemma:2S-orig}
	Let $\T > 0$. Any complete calculation $\RR$ of the red-white pebble game on a CDAG $G$ using at most $\S$ red pebbles
	is associated with a $\ST$-partition of the CDAG such that
	\[ Q^{\RR} \geq \T \cdot (h - 1),\]
	where $Q^{\RR}$ is the number of applications of rule (R1) in the game and $h$ is the number of subsets in the partition.
\end{lemma}

In particular, an upper bound $U$ on the size of a $\ST$-bounded set directly translates into
a data movement lower bound: $\Q(G) \ge \T \cdot \left(\left\lceil \frac{\card{V\setminus I}}{U}\right\rceil - 1\right)$.
This bound can actually be improved to:
\begin{equation}\label{equation:2S-improved}
\Q(G) \ge \T \cdot \left\lfloor \frac{\card{V\setminus I}}{U}\right\rfloor.
\end{equation}
Indeed, the proof of Lemma~\ref{lemma:2S-orig} establishes a correspondence between the number of sets in a $\ST$-partition
of $G$ and a partition of an execution of the game in segments containing exactly $\S$ loads, except maybe the last one.
The subtraction by one is due to this last segment. When $\card{V\setminus I} / U$ is an integer, the segment contains exactly $\S$ loads and
thus the subtraction is not necessary.

We actually want to be able to compute lower bounds for CDAGs in which no vertices are tagged as input
(this is particularly useful when doing decomposition, see Sec. \ref{sec:decomp})
The following lemma (Lemma~\ref{lemma:2S}) establishes such a result. The main idea is as follows: we tag some vertices
as input, getting a new CDAG on which Lemma~\ref{lemma:2S-orig} applies and gives some bound.
The additional \IO cost is at most the number of input vertices that were added, so we get a lower bound for the original
CDAG by subtracting this number to the bound. See~\cite{elango-spaa2014} for a complete proof.

\begin{definition}[\Sources]
	Let $G=(V,E)$ be a DAG, $P\subseteq V$ be a vertex set in $G$.
	The \sources of $P$ are the vertices of $P$ with no predecessors in $P$.
	Formally,~\\
		\[\msources{P} = \{v \in P,\ \nexists u \in P, (u,v) \in E\}\]
\end{definition}

\begin{lemma}[$\ST$-Partitioning \IO lower bound, no input case~\cite{elango-spaa2014}]\label{lemma:2S}
	Let $\S$ be the capacity of the fast memory, let $G=(V,E,\emptyset)$ be a CDAG, and let $h$ be the minimum number
	of subsets in a $\ST$-partition of $G_I = (V, E, I=\msources{V})$ for some $\T > 0$.
	Then, the minimum \IO for $G$ satisfies:
	\[\Q(G) \geq \T \cdot  (h - 1) -\card{\msources{V}}.\]
\end{lemma}

\subsection{Using projection to bound the cardinality of $\K$-bounded sets}
\label{ssec:projection}

The key idea behind the automation of data movement lower bound computation is the use of geometric inequalities through an appropriate program representation. 
Vertices of a CDAG are mapped to points in a multidimensional geometric space $\mathcal{E} \simeq \Z^d$ through some mapping $\rho$ (where dimensions are typically loop indices), and regular data dependencies in the CDAG are represented as projections on a lower-dimensional space.

The condition ``set of vertices $P \subset V$ is $\K$-bounded'' in the CDAG corresponds to a condition of the form ``the size of the projections of $\rho(P)$ in $\mathcal{E}$ is bounded by $\K$''. 
Finding a bound on the size of a $\K$-bounded set in a CDAG can thus be reduced to: finding a bound on the size of a set $E$ in a geometric space, given cardinality bounds on some of its projections. 
This correspondence is developed in Sec.~\ref{sec:partition}. 
In this section, we take it for granted and only introduce the mathematical notations and results.

There exist inequalities for doing exactly what we need, namely the discrete Brascamp-Lieb inequality,
introduced by Christ et al.~\cite{Demmel2013TR} as a discrete analogue to the one established by Brascamp and Lieb for metric spaces~\cite{brascamp76}.

\begin{theorem}[Brascamp-Lieb inequality, discrete case~\cite{Demmel2013TR}]
	\label{thm:bl}
	Let $d$ and $d_j$ be nonnegative integers and $\phi_j : \Z^d \mapsto \Z^{d_j}$ be group homomorphisms for $1 \leq j \leq m$.
	Let $0 \leq s_1, s_2, \dots, s_m \leq 1$.
	Suppose that:
	\begin{equation} 
		\rk{H} \leq \sum_{j=1}^m s_j\cdot\rk{\phi_j(H)} \text{ for all subgroups $H$ of $\Z^d$}
		\label{eq:bl1}
	\end{equation}
	Then:
	\begin{equation}
		\label{eq:bl2}
		\card{E} \leq \prod_{j=1}^m \card{\phi_j(E)}^{s_j} \text{for all nonempty finite sets $E \subseteq \Z^d$}.
	\end{equation}
\end{theorem}

A special case is when the $\phi_j$ are the canonical projections on $\Z^{d-1}$
along basis vectors, and $s_1 = \dots = s_m = \frac{1}{d-1}$,
giving a bound of the form $\card{E} \le \prod_{i=1}^d \card{\phi_j(E)}^{1/(d-1)}$.
Fig.~\ref{fig:loomis2d} illustrates this special case in two dimensions.

The issue, when trying to apply Theorem~\ref{thm:bl}, is that $(\ref{eq:bl1})$ has to be true for \emph{all} subgroups, which can obviously be quite
difficult to establish.
However, as all the coefficients are integers and bounded by $d$, the number of distinct inequalities in (\ref{eq:bl1}) is bounded.
	The set of admissible $s_j$ is thus a (convex) polyhedron
and it has been shown~\cite{christ2015holder} that the polyhedron defined by these inequalities is computable.
The algorithm is actually combinatorial and quite complex, so we do not use it in our present work.
	Instead, we use the following result, which restricts the set of subgroup for which $(\ref{eq:bl1})$ has to be verified.

        \begin{definition}[Lattice of subgroups]
          \label{def:lattice}
	The \emph{lattice of subgroups} generated by subgroups $H_1, H_2, \dots, H_m$ of a group $G$ is the closure of $\{H_1, H_2, \dots, H_m\}$ under group sum and intersection.
\end{definition}

\begin{lemma}[Lattice of subgroups in Brascamp-Lieb\cite{valdi}]
  \label{lm:lattice}
	Theorem~\ref{thm:bl} holds with the weaker condition:
	\begin{equation} 
		\rk{H} \leq \sum_{j=1}^m s_j\cdot\rk{\phi_j(H)} \text{ for all subgroups $H \in \latphi$} \tag{\ref{eq:bl1}b}
		\label{eq:bl1b}
	\end{equation}
	where $\latphi$ is the lattice of subgroups generated by $\Ker (\phi_1), \Ker(\phi_2), \dots, \Ker(\phi_m)$.
\end{lemma}

The subgroup lattice is not necessarily finite, so this does not give a tractable algorithm,
but this will be sufficient in most cases, as loop nests are usually of quite small dimensions, and data dependencies
are not too complex. In our practical implementation, we use
a time-out. We add projections: the more projections the tighter the 
bound. Each time we add a projection we update the lattice of subgroups and check for Conditions~\ref{eq:bl1b} to be satisfied. The process (adding projections) stops if we reach the time-out. This does not mean that the algorithm fails, rather that the bound will potentially be less tight (cf. Sec.~\ref{sec:complete}).

A situation that often occurs is when the subgroups $\Ker(\phi_j)$ are linearly independent: in this case there is no need to compute the generated lattice of subgroups,
simply testing on each kernel individually for Conditions~\ref{eq:bl1} is sufficient (see proof in \cite{Demmel2013TR}, Sec.~6.3).

\paragraph{Choosing $s_j$'s}
The goal is to solve the following problem:
``\emph{Given a set of projections (group homomorphisms in $\Z^d$) $\phi_1,\dots,\phi_m$ and a constant $\K$, find an upper bound (as tight as possible)
  on the cardinality of a set $E \subset \Z^d$ satisfying $\card{\phi_j(E)} \leq \K$.}''

For any coefficients $s_1,\dots,s_j$ satisfying $(\ref{eq:bl1})$, Theorem~\ref{thm:bl}
gives the following bound on $\card{E}$:
\[\card{E} \le \prod_{j=1}^m\card{\phi_j(E)}^{s_j}\le \prod_{j=1}^m \K^{s_j} = \K^{\sum_j s_j}.\]

To get a bound as tight as possible on $\card{E}$, we want to minimize the right-hand side of this inequality.
This amounts to minimizing $\sum_j s_j$ while satisfying the constraints in $(\ref{eq:bl1})$.
Since this constraints are linear inequalities, the optimal choice for $s_j$'s can be obtained by a linear solver.

In the special case where the $\phi_j$'s are orthogonal projections along basis vectors (and $m=d$), 
kernels are linearly independent and the linear program is:
\begin{align*}
	\text{Minimize\ \  } \sum_j s_j \text{\ \ \ \ \ \ \  s.t.\ \  }	\forall 1\le i\le d,\ 1 \le \sum_{j \ne i} s_j
\end{align*}
and its solution is, as expected, $s_1 = \dots = s_d = \frac{1}{d-1}$.

\subsection{A compact representation of the CDAG: the \DFG}
 \label{sec:DFG}

\begin{figure}[h]
\captionsetup[subfigure]{justification=centering}
 \centering
 \hfill
	\begin{minipage}{0.39\textwidth}
 \begin{subfigure}{\textwidth}
    \footnotesize
    \lstset{language=C,basicstyle=\ttfamily,escapeinside={<@}{@>}}
    {\begin{lstlisting}
for (<@$0\le t<M$@> and <@$0\le i<N$@>)
  if (t==0): S[0,i]=A[i]*C[0];
  else: S[t,i]=S[t-1,i]*C[t];
    \end{lstlisting}}
		\caption{Single assignment form}
 \end{subfigure}
 \centering
	\begin{subfigure}{\textwidth}
		\centering
		\includegraphics[width=0.5\textwidth]{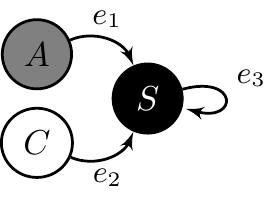}
		\caption{DFG\label{fig:ex-2d-dfg}}
	\end{subfigure}
        \end{minipage}
	\hfill
	\begin{minipage}{0.6\textwidth}
 \begin{subfigure}[t]{\textwidth}
 \footnotesize
   $D_{A}=[N]\rightarrow\left\{A[i]:\ \ 0\le i<N\right\}$ \\
   $D_{C}=[N]\rightarrow\left\{C[t]:\ \ 0\le t<M\right\}$ \\
   $D_{S}=[M,N]\rightarrow\left\{S[t,i]:\ \ 0\le t<M\ \ \wedge\ \ 0\le i<N\right\}$ \\
   $\card{D_{S}}=MN$
   \caption{Node domains}
	 \label{fig:ex-2d-domains}
 \end{subfigure}

  \vspace{.5cm}

 \begin{subfigure}[t]{\textwidth}
 \footnotesize
   $R_{e_1}=[N]\rightarrow\left\{A[i]\rightarrow S[0,i]:\ \ 1\le i<N\right\}$ \\
   $R_{e_2}=[M,N]\rightarrow\left\{C[t]\rightarrow S[t,i]:\ \ 0\le t<M\ \ \wedge\ \ 0 \le i < N\right\}$ \\
   $R_{e_3}=[M,N]\rightarrow\left\{S[t,i]\rightarrow S[t+1,i]:\ \ 0\le t<M-1\ \ \wedge\ \ 0\le i<N\right\}$\\
   \caption{Edge relations}
	 \label{fig:ex-2d-relations}
 \end{subfigure}
	\end{minipage}
        	\hfill

 \caption{\dfg for Example~1}
 \label{fig:ex-2d-2}
\end{figure}

	A CDAG (see Fig.~\ref{fig:ex-2d-cdag}) represents a single dynamic execution of a program, and can be very large.
	To be able to analyze programs of realistic size with reasonable resources, we
	use a compressed representation called \DFG (\dfg).
	Another advantage of such a representation is that it is \emph{parametric},
	i.e.~a single \dfg can represent CDAGs of different sizes, depending on program parameters.
	A \dfg represents an \emph{affine} computation, which is the class of
	programs that can be handled by the \emph{polyhedral model}~\cite{polyhedron}.
	We use the terminology and syntax from the ISL library~\cite{verdoolaege2010isl},
	and illustrate them with the example of Fig.~\ref{fig:ex-2d}.
        Formal definitions can be found in the manual~\cite{isl-manual}.
        
        \paragraph{Vertex domains} As one can see on Fig.~\ref{fig:ex-2d-cdag}, to each loop is associated a ``geometric'' space dimension ($t$ and $i$ here) so that each \emph{vertex of the CDAG} lives in a multidimensional iteration space, its \emph{domain}, that can be algebraically represented as a union of parametric $\Z$-polyhedra bounded by affine inequalities.

	A domain is an \emph{ISL set} for which standard operations (union, intersection, difference,\dots) are available, as well as a \emph{cardinality} operation (denoted $\card{D}$). 
        As an example (see Fig.~\ref{fig:ex-2d-domains}), the \emph{domain} $D_{S}$ of statement $S$ is a $\Z$-polyhedron with parameters $M$ and $N$ made up of all integer points $(t,i)$ such that $0\le t<M$ and $0\le i<N$.
        The number of points in this set (cardinality) is $\card{D_{S}}=MN$.
        Note that the space within which all the points of a statement ($S$ here) live is identified with the name of the statement, using the notation $S[t,i]$.
        
        \paragraph{Edge relations} A set of \emph{edges} of the CDAG is represented using a \emph{relation} (ISL map), which is a set of pairs between two spaces, from the \emph{domain} space to the \emph{image} space.
        As an example (see Fig.~\ref{fig:ex-2d-relations}), the data flow from statement $S[t,i]$ (definition of $A[i]$ in $S$) to statement $S[t+1,i]$ (use of $A[i]$ in $S$) is represented using the relation $R_{e_3}$.

	In addition to standard set operations, ISL can compute the transitive closure of a relation,
	denoted $R^*$.
	Binary relations are also supported: image of a domain $D$ through a relation $R$ (denoted $R(D)$),
	and composition of two relations $R_1$ and $R_2$, denoted $R_1\circ R_2$ (this is left composition, going the opposite way from usual functional notation). 
	Composition restricts the image domain of the resulting relation to points where the composition relation makes sense:
	$\domain{R_1\circ R_2} = R_1^{-1}\left(\image{R_1}\cap\domain{R_2}\right)$,
	$\image{R_1\circ R_2} = R_2\left(\image{R_1}\cap\domain{R_2}\right)$.
	As with domains, we will sometimes manipulate unions of such relations.

\paragraph{A \DFG (\dfg)} A \dfg is a graph $G = (\V,\E)$.
Each vertex $S\in\V$ of the graph represents a (static) statement or an input array of the
program.
Each vertex $S$ is associated with a parametric iteration domain $D_S$ and a list of enclosing loops (empty for input arrays).
Each edge $d = (S_a, S_b)\in\E$ represents a flow dependency between statements or input
arrays.
Each edge is associated with an affine relation $R_d$ between the
coordinates of the source and sink vertices.
The \dfg is a compact (exact) representation of the dynamic CDAG where a single vertex/edge of the \dfg represents several vertices/edges of the dynamic CDAG.
While all the reasoning and proofs can be done by visualizing a CDAG, the actual heuristic described in this paper manipulates its compact representation, allowing to translate graph methods~\cite{elango-popl2015} into geometric reasoning.
Fig.~\ref{fig:ex-2d-dfg} shows the \dfg for our simple stencil code.

\paragraph{\dfg-paths}
A fundamental object in our lower bound analysis is a \dfg-path, which is simply a directed path in a \dfg.
The relation $R_p$ of a \dfg-path $p = (e_1,\dots,e_k)$ is the composition of the relations of
its edges: $R_p = R_{e_1}\circ \dots\circ R_{e_k}$.
We are only interested in two specific types of \dfg-paths, depending on their relation:
\begin{itemize}
\item \emph{chain circuits}, which are cycles from one \dfg-vertex $S$ to itself, such that the path relation $R_p$ is a translation $S[\vec{x}]\rightarrow S[\vec{x} + \vec{b}]$.
\item \emph{broadcast $S_a,S_b$-paths}, which are elementary paths (from a $S_a$ to $S_b$ -- $S_b$ possibly equal to $S_a$) in which all \dfg-edges but the first one are injective edges, such that the inverse of the corresponding relation $R_p$ is an affine function $S_b[\vec{x}]\rightarrow S_a[A\cdot\vec{x} + \vec{b}]$, where $A$ is not full-rank.
\end{itemize}

In Fig.~\ref{fig:ex-2d-2}, path $p = (e_3)$ is a chain circuit, going from $S$ to itself with translation vector $\vec{b} = (1, 0)$. 
Path $p' = (e_2)$ is a broadcast path, with relation $R_{p'} = R_{e_2} = \{C[t]\rightarrow S[t,i]:\ \ 0\le t<M\ \ \wedge\ \ 0 \le i < N\}$. 
The inverse relation is the linear function $\vec{I} \mapsto A \cdot \vec{I} + \vec{b}$, with $A = (1~~0), \vec{b} = (0)$. 
The kernel of $A$ is $\{(0, i), i \in \R\}$.

\section{CDAG Decomposition}
\label{sec:decomp}

To derive data movement lower bounds for a complex program, it is essential to be able to decompose it into subregions for which we can compute lower bounds, and then sum the complexity for each subregion. 
The \emph{no recomputation} condition is necessary for such a decomposition. 
Under this hypothesis, it is quite straightforward to see that a decomposition into disjoint subregions is sufficient. 
In this section, we provide a more general decomposition lemma, using the fact that vertices of a subregion that will not be counted as loads can also be part of another subregion. 
We then explain how it is applied on the \dfg representation, distinguishing two cases: 
combining a fixed number of program regions (see example in Fig.~\ref{fig:fw-2d}); 
and summing over all iterations of a loop (see example in Fig.~\ref{fig:re1d}), which amounts to combining an unbounded (parametric) number of program regions. 
We stress that the CDAG partitioning method (sections~\ref{ssec:partitioning} and \ref{ssec:projection}) and the CDAG decomposition method (this section) are very much different and not related.

\subsection{Non-disjoint Decomposition Lemma}
\begin{definition}[sub-CDAG, no-spill set]
  \label{def:may-spill}
	Let $G = (V, E, I)$ be a CDAG, and $V_i \subset V$.
	The sub-CDAG $G_{|V_i}$ of $G$ is the CDAG with vertices $V_i$, edges $E_i = E\cap (V_i \times V_i)$ and input vertices $I_i = I \cap V_i$.

	The \emph{no-spill set} of $G_{|V_i}$ is the subset of vertices of $V_i\setminus I_i$ with either:
	\begin{enumerate}
		\item no outgoing edges in $E_i$, or 
		\item no incoming edges in $E_i$ and at most one outgoing edge in $E_i$
	\end{enumerate}

	The \emph{may-spill set} of $G_{|V_i}$ is the complementary of its no-spill set in $V_i$.
\end{definition}

\begin{lemma}[CDAG decomposition] \label{lemma:composition}
	Let $G = (V, E, I)$ be a CDAG.
	Let $V_1,V_2,\dots,V_k$ be subsets of $V$ such that for any $i\neq j$, the may-spill sets of $G_{|V_i}$ and $G_{|V_j}$ are disjoint.
	Then
	\[\Q(G) \ge \sum_{i=1}^k \Q(G_{|V_i}).\]
\end{lemma}

\begin{proof}
	Let $\RR$ be an optimal $\S$-RW-game on $G$, with cost $Q = \Q(G)$.
	For all $i$, we denote $Q_i$ the cost of $\RR$ restricted to $V_i$, that is the number of applications
	of rule (R1) on vertices in \emph{the may-spill set of} $V_i$.
	Since the may-spill sets are pairwise disjoint, clearly $Q \geq \sum_{i=1}^k Q_i$.
	For all $i$, we will build from $\RR$ a valid game $\RR_i$ for $G_{|V_i}$ with cost $Q_i$.
	This will show that $\Q(G_{|V_i}) \leq Q_i$, from which follows 
	$\sum_{i=1}^k \Q(G_{|V_i}) \le \sum_{i=1}^k Q_i \le Q = \Q(G)$, establishing the result.

	To build the game $\RR_i$ from $\RR$, we proceed as follows:
	\begin{enumerate}
		\item Remove every move involving vertices outside $V_i$.
		\item For no-spill vertices in $V_i$ without successors, remove every application of rule (R1) and (R3).
		\item For no-spill vertices $v$ in $V_i$ with no predecessors and one successor $w$, move the single application
			of rule (R2) just before application of (R2) on $w$, add a (R3) move for $v$ just after this point and remove all subsequent (R1) and (R3) moves on $v$.
	\end{enumerate}
	It is clear that after step 1, we have a valid game for $G_{|V_i}$.
	Indeed conditions for (R1) and (R3) are trivially preserved, and since we kept all (R2) moves on vertices in $V_i$,
	there will always be the necessary red pebbles to apply (R2).
	Step 2 also gives a valid game, since no (R2) moves can depend on such a vertex having a red pebble.
	Step 3 is also a valid transformation because since $v$ does not have any predecessor in $V_i$ (and is not an input vertex),
	it can be activated via (R2) and any given point. Therefore activating it just when it is needed, and applying (R3) just after is valid.

	This preserves all (R1) moves on the may-spill set of $G_{|V_i}$ and removes all (R1) moves on its no-spill set, thus the cost of $\RR_i$ is indeed $Q_i$.
\end{proof}

\begin{figure}[h]
\captionsetup[subfigure]{justification=centering}
 \begin{subfigure}[c]{0.4\textwidth}
   \lstset{language=C,basicstyle=\ttfamily,escapeinside={<@}{@>}}
    {\footnotesize\begin{lstlisting}
  for(t=0; t<M; t++) {
      s = 0;
      for(i=0; i<N; i++)
S1:       s += A[j];
      for(i=0; i<N; i++)
S2:       A[j] += s;
  }
    \end{lstlisting}}
	 \caption{\label{fig:re1d-code}Code}
 \end{subfigure}
 \begin{subfigure}[c]{0.3\textwidth}
	 \includegraphics[width=0.8\textwidth]{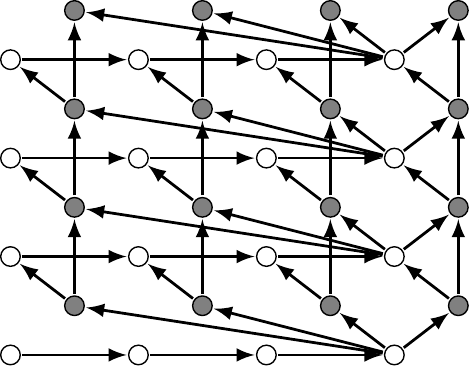}
	 \caption{\label{fig:re1d-cdag}CDAG for M=4, N=4. White vertices correspond to S1, gray vertices to S2.}
 \end{subfigure}

 \begin{subfigure}[c]{0.7\textwidth}
	 \includegraphics[width=\textwidth]{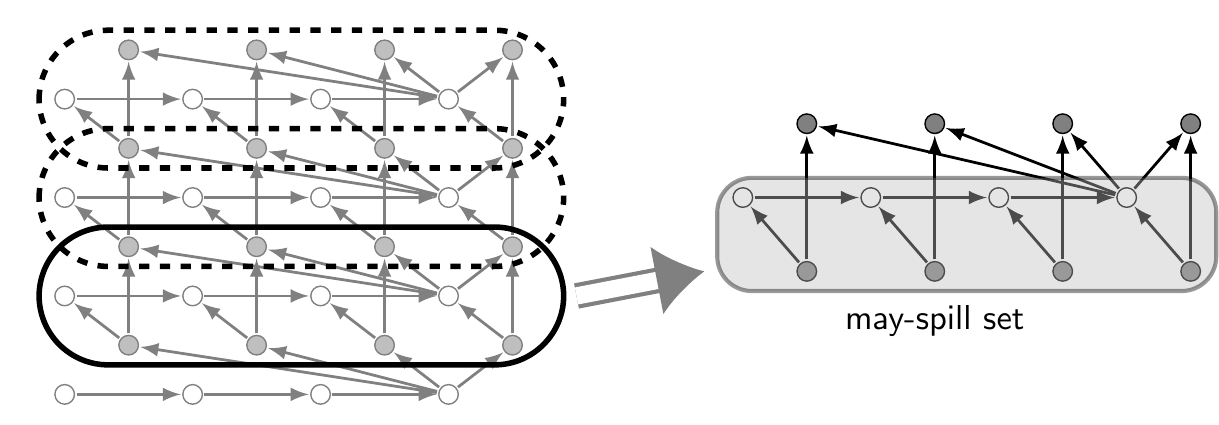}
	 \caption{\label{fig:re1d-decomp}Decomposition of the CDAG}
 \end{subfigure}
	\caption{\label{fig:re1d}Example 2}
\end{figure}

\tool implements two different mechanisms that make use of the non-disjoint decomposition lemma.
The basic one (bounded combination -- Sec.~\ref{ssec:maxcover}) simply decomposes the CDAG into a bounded number of sub-CDAGs (e.g.~corresponding to different sub-regions of the code), computes the corresponding \IO complexities, and combines them.
The more complex one (loop parametrization -- Sec.~\ref{ssec:loop}), decomposes the CDAG into an unbounded number of sub-CDAGs by ``slicing'' the iteration space of a loop nest.
\tool combines the two mechanisms.
The following example illustrates the decomposition lemma for loop parametrization.

\paragraph{Illustrating example}
Consider Example 2 on Fig.~\ref{fig:re1d}. 
The CDAG can be decomposed into $M-1$ identical subgraphs, as shown on Fig.~\ref{fig:re1d-decomp} (each subgraph $G_{|V_t}, t = 1,\dots,M-1$ corresponds to iteration $t$ of the loop enclosing $S_1$, and iterations $t-1$ and $t$ of the loop enclosing $S_2$). 
On each of these subgraphs, the may-spill set contains the two ``bottom'' rows (because vertices in the ``top'' row have no successor in the sub-CDAG). 
Thus the may-spill sets of these subgraphs are pairwise disjoint and the \IO for the whole CDAG is greater than the sum of the individual \IO for each subgraph by Lemma~\ref{lemma:composition}.

On each subgraph $G_{|V_t}$, the wavefront method (Sec.~\ref{sec:wf}) can be applied, giving a lower bound on \IO of $\Q(G_{|V_t}) \ge N - \S$.
As the may-spill set of the different subgraphs do not intersect, the individual complexities can be summed over $t = 1,\dots,M-1$, providing a lower bound for the whole CDAG:
\[ \Q(G) \ge (M-1)(N-\S). \]

\subsection{Bounded combination}
\label{ssec:maxcover}
The main procedure of \tool (Sec.~\ref{ssec:mainloop}) selects a bounded set of (possibly overlapping) sub-CDAGs and computes their individual complexities.
The objective of Alg.~\ref{alg:maxcover} is to combine (sum) as many non-interfering (disjoint may-spill sets) complexities as possible.
It does so using a greedy approach:
Assume there are two sub-CDAGs both with a ``high'' complexity but with non-disjoint may-spill sets.
Alg.~\ref{alg:maxcover} will select the one with the highest complexity, recompute the complexity of the second after removing the intersecting part, and then sum them up.
The overall set of sub-CDAGs is iteratively processed this way (and the complexities summed-up) until empty or negligible complexities remain.
The comparison (what is ``higher'') is done using \emph{instances of parameter values}, simply evaluating the corresponding symbolic expressions.
It should be emphasized that the final bound is a valid lower bound for \emph{any} parameter values, the instances of parameter values are only used for heuristics.
 
\begin{algorithm}[h]
	\small

\Fn{\maxcover}{
  \Input{ A \dfg $G$, an instance $\i$, a set of complexities $\Ginterf$}
  \Output{ A combined complexity $\Q^\i$}
  $\Q^\i=0$\;
  Let $G'$ be a copy of $G$\;
  \While{$\Ginterf\neq \emptyset$}{
    \textbf{let} $Q$ such that $Q(\i)=\max_\Ginterf Q(\i)$\;
    $\Ginterf = \Ginterf - \{Q\}$\;
    \lIf{$Q(\i)=0$}{\textbf{return} $\Q^\i$}
    \uIf{$G'\cap Q.\mayspill\neq \emptyset$}{
      Recompute $Q$ assuming CDAG $G'$\;
      $\Ginterf = \Ginterf \cup \{Q\}$\;
    }
    \Else{
      $\Q^\i:=\Q^\i+Q$\;
      $G'=G'-Q.\mayspill$
    }
  }
  \Return $\Q^\i$
}

        \vspace{\baselineskip}
	\caption{\label{alg:maxcover}Summing lower bound expressions by removing interferences}
\end{algorithm}

\begin{figure}[h]
\captionsetup[subfigure]{justification=centering}
	\begin{minipage}{0.5\textwidth}
  \begin{subfigure}[b]{\textwidth}
    \centering
    \lstset{language=C,basicstyle=\ttfamily,escapeinside={<@}{@>}}
    {\footnotesize\begin{lstlisting}
<@\textbf{Parameters}@>: N; 
<@\textbf{Input}@>: A[N]; <@\textbf{Output}@>: A[N];
for(k=0;k<N;k++)
  for (i=0; i<N; i++) 
    A[i] = f(A[i],A[k]);
    \end{lstlisting}}
    \caption{C-like code\label{fig:fw-2d-code}}
  \end{subfigure}
  
  \begin{subfigure}[b]{\textwidth}
    \centering
    \lstset{language=C,basicstyle=\ttfamily,escapeinside={<@}{@>}}
    {\footnotesize\begin{lstlisting}
<@\textbf{Parameters}@>: N; 
<@\textbf{Input}@>: A[N]; <@\textbf{Output}@>: <@$S_{N-1}$@>[N];
for (<@$0\le k<N$@> and <@$0\le i<N$@>)
  <@\tikzmark{fw1}@>if (k==i==0): <@$S_{0,i}$@> = f(A[0],A[0]);
  else if (k==0): <@$S_{0,i}$@> = f(A[i],<@$S_{0,0}$@>);<@\tikzmark{fw2}@>
  else if (i<=k): <@$S_{k,i}$@> = f(<@$S_{k-1,i}$@>, <@$S_{k-1,k}$@>);
  else if (i>k): <@$S_{k,i}$@> = f(<@$S_{k-1,i}$@>, <@$S_{k,k}$@>);
    \end{lstlisting}}
		\caption{Corresponding single assignment form\label{fig:fw-2d-sa}}
	\end{subfigure}
	\begin{tikzpicture}[overlay, remember picture]
		\node [above left=5pt and 0pt of pic cs:fw1] (fw1){};
		\node [below right=1pt and 2pt of pic cs:fw2] (fw2){};
		\draw [white,fill=white, opacity=0.6] (fw1) rectangle (fw2);
	\end{tikzpicture}
	\end{minipage}
  \begin{subfigure}[c]{0.4\textwidth}
    \includegraphics[width=\textwidth]{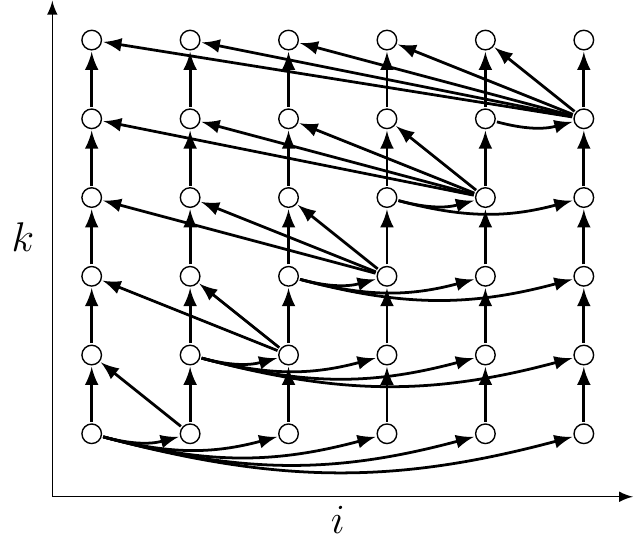}
    \caption{Corresponding CDAG for N=5. Input nodes \texttt{A[N]} are omited.\label{fig:fw-2d-cdag}}
  \end{subfigure}

  \begin{subfigure}[c]{0.7\textwidth}
    \includegraphics[width=0.4\textwidth]{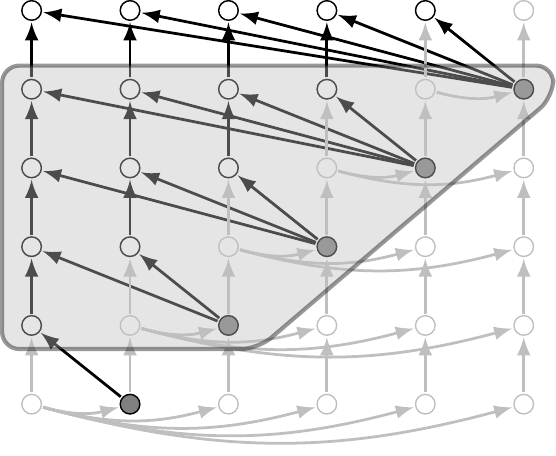}
		\hfill
    \includegraphics[width=0.4\textwidth]{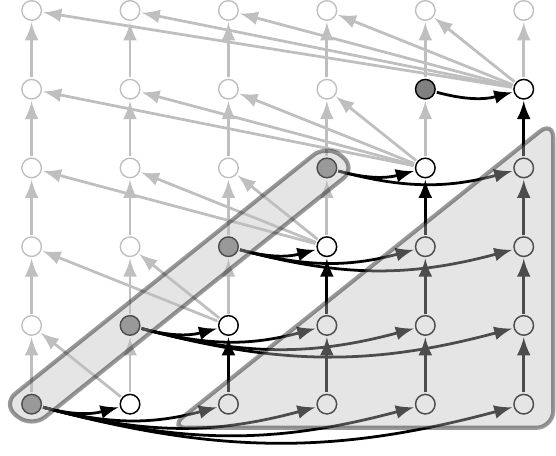}
    \caption{\label{fig:fw-2d-decomp}Decomposition into two non-interfering sub-CDAGs. Sources are in gray. May-spill sets are encircled.}
  \end{subfigure}
	\caption{\label{fig:fw-2d}Example~3}
\end{figure}
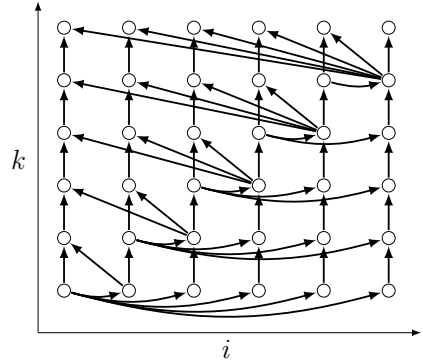
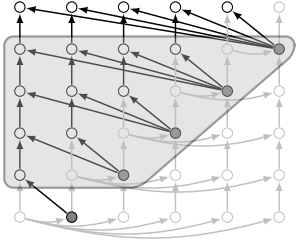

Let us have a look at the example on Fig.~\ref{fig:fw-2d}.
In the original code (\ref{fig:fw-2d-code}), notice that \texttt{k} is the outer loop index,
meaning that \texttt{A[k]} will have been modified either in the current loop iteration or the previous one
depending on the order between \texttt{i} and \texttt{k}
(Floyd-Warshall exhibits the same pattern, with three loops instead of two).
This is made clear in the single-assignment form (\ref{fig:fw-2d-sa}),
and can be visualized in the CDAG representation (\ref{fig:fw-2d-cdag}). The dependences on input values
are grayed out in (\ref{fig:fw-2d-sa}) and omitted in (\ref{fig:fw-2d-cdag}), and we will
ignore them in the discussion to keep the explanations simple.

Considering only the statement vertex $S$ in the \dfg, the dependency analysis gives the following relations:

\begin{align*}
	R_1 &= \{S[k-1,i]\, \rightarrow S[k,i]:\ \ 1\le k<N\ \ \wedge\ \ 0 \le i < N\} \\
	R_2 &= \{S[k-1,k]\rightarrow S[k,i]:\ \ 1\le k<N\ \ \wedge\ \ 0 \le i < k\} \\
	R_3 &= \{S[k,k]\ \ \ \rightarrow\ \ \ S[k,i]:\ \ 0\le k<N\ \ \wedge\ \ k < i < N\} \\
\end{align*}

The image domains of $R_2$ and $R_3$ provide a natural decomposition of the CDAG into two non-interfering sub-CDAGs,
as shown in (\ref{fig:fw-2d-decomp}).
On each part, the pattern is similar to that of Example~1 on page \pageref{fig:ex-2d} the geometric approach gives a lower bound (omitting lower order terms) $\Q(G_i) \ge \frac{N^2}{2S}$.
Since they do not interfere, Alg.~\ref{alg:maxcover} will return their sum $Q(G) \ge \frac{N^2}{S}$, independently of the parameter instance.

\subsection{Loop parametrization}
\label{ssec:loop}
As done on the example above, \tool can compute the \IO complexity of some inner loop nests of a bigger enclosing loop nest and sum them.
To this end, our scheme performs what we call \emph{loop parameterization}.
Loop parameterization considers each individual sub-CDAGs where the outermost indices are fixed (our algebraic formulation obviously allows to consider such indices as parameters without the need to explicitly enumerate them) enriched by their input vertices.
Taking the notations
\[\closin{V_i}=V_i\cup \In{V_i}\]
parametrizing the outer ``\texttt{t}'' loop with $t=\Omega$ (with $\Omega$ a parameter) allows to compute $\Q^\Omega$, a (parametric) lower bound for each individual value of $\Omega = 1,\dots,M-1$ (The sub-CDAG for $\Omega=0$ does not have the same pattern so it is ignored), and combine them
\[\Q=\sum_{1\le \Omega< M}\Q^\Omega=\sum_{1\le \Omega< M} Q_{|\closin{\{v\in V,\ t=\Omega\}}}=\sum_{1\le \Omega< M} N-S=(M-1) \cdot (N-S).\]
In more complex cases, the parametric bound can depend on the outer loop parameter $\Omega$, and we use formulas for sum of polynomials.

In ISL terms, this is done by making the outer loop index a parameter. Here the original domain
\[D_{S_1} = [M,N] \rightarrow \{S_1[t,i] : 0 < t < M \ \ \wedge\ \ 0 < i < N\}\]
becomes
\[D_{S_1}^{\Omega} = [M, N, \Omega] \rightarrow \{S_1[t, i]: t = \Omega \ \ \wedge\ \ 0 < t < M \ \ \wedge\ \ 0 < i < N\}. \]

The corresponding parts of the algorithms are highlighted in Algorithm~\ref{alg:mainloop} on page~\pageref{alg:mainloop}.

\section{$\K$-partition bound derivation}
\label{sec:partition}

In this section, we explain how to apply the geometrical reasoning of Sec.~\ref{ssec:projection} on a CDAG $G = (V,E)$, using its compact representation as a \dfg. 
We also present, in~\ref{sssec:sum-trick}, a generalization of one of the techniques introduced in~\cite{dongarra-2008,langou-gemm-14,smith-gemm-17,smith-gemm-19} that these authors used to derive a tighter lower bound for matrix multiplication.

To apply Lemma~\ref{lemma:2S} on $G$, we need to find a lower bound on the minimum number of subsets $h$ in any $\K$-partition of $G$.
The general reasoning is as follows:
\begin{enumerate}
	\item Embed $V$ in a geometric space through a map $\rho : P \subseteq V \mapsto E \subseteq \Z^d$, such that two disjoint subsets
		of $V$ are mapped to disjoint subsets of $\Z^d$. We have
		$$\card{\rho(P)} \leq \card{P}.$$
	\item Use the \dfg representation to find a subset $V' \subseteq V$ and a set of projections (group homomorphisms) 
		$\phi_1,\dots,\phi_m$ with the property that:
		\begin{equation}
			\label{eq:proj-bnd}
			\text{Any  $\K$-bounded set $P \subseteq V' \setminus \msources{V'}$ satisfies} \card{\phi_j(\rho(P))} \leq \K.
		\end{equation}
	\item Using Theorem~\ref{thm:bl}, derive an upper bound $U$ on $\card{\rho(P)}$ for any $\K$-bounded $P$. This 
		provides a lower bound $\left\lceil \frac{\card{V'\setminus \msources{V}}}{U} \right\rceil$ on the number $h$ of disjoint $\K$-bounded sets in $V'\setminus\msources{V'}$.
\end{enumerate}

\subsection{Geometric embedding, \dfg-paths and projections}
\label{ssec:embedding}

Let $S_k$ be some fixed \dfg-vertex (corresponding to one program statement).
Let $Q_1,\dots,Q_m$ be \dfg-paths all ending in $S_k$, with a common image domain $D_{k} = \{S_k[i_1,\dots,i_d]: \dots\}$.

The embedding $\rho$ is defined as:
\[\rho(P) = \{(i_1,\dots,i_d)\ \  |\ \  S_k[i_1,\dots,i_d] \in P\}, \]
that is vertices corresponding to statement $S_k$ are mapped to their corresponding $d$-dimensional point, and other vertices are ignored.

\begin{definition}[embedded projections]
  \label{def:proj}
For a given path $Q$ with relation $R_Q$, the geometric projection $\phi_Q$ is defined as follows:
\begin{itemize}
	\item If the path is a broadcast path with $R_Q = \{S_j[j_1,\dots,j_{d'}] \rightarrow S_k[i_1,\dots,i_d]: \dots \}$ for some statement $S_j$ (not necessarily $\ne S_k$),
		then the projection is directly given by the path relation $\phi_Q(i_1,\dots,i_d) = (j_1,\dots,j_{d'})$.
	\item If the path is a chain circuit with $R_Q = \{S_k[i_1,\dots,i_d] \rightarrow S_k[i_1+\delta_1,\dots,i_d+\delta_d]: \dots \}$,
		then the projection is the orthogonal projection on the hyperplane in $\Z^d$ defined by orthogonal vector $\delta = (\delta_1,\dots,\delta_d)$.
		Its explicit formulation can be computed but is not needed here.
\end{itemize}
\end{definition}

In the case of a broadcast, it is straightforward that $\phi_Q$ satisfies (\ref{eq:proj-bnd}), because $R_Q^{-1}(P)$
is included in $\In{P}$ for any $P \subseteq V$, so $\card{\phi_Q(\rho(P))} \leq \card{\In{P}} \leq \K$ for any $\K$-bounded $P$.

In the case of a chain circuit, let us call $I_Q(P) = R_Q^{-1}(P) \setminus P$. This is basically the
\inset of $P$ restricted to edges corresponding to \dfg-path $Q$, so $I_Q(P) \subset \In{P}$.
The projection $\phi_Q$ associates one point to each straight line directed by $\delta$.
Since $P \subseteq V \setminus \msources{V}$, there is at least one point in $\phi_Q(\rho(P))$ for every nonempty chain in $P$,
and $\card{\phi_Q(\rho(P))} \leq \card{\In{P}} \leq \K$.

\paragraph{Example} Consider paths $p_1 = (e_2)$ and $p_2 = (e_3)$ in Fig.~\ref{fig:ex-2d-2}.
$p_1$ is a broadcast path with relation $\{C[t] \rightarrow S[t,i]\}$, so the corresponding projection is $\phi_1(t, i) = (t)$.
$p_2$ is a chain path with relation $\{S[t,i] \rightarrow S[t+1,i]\}$, so the corresponding projection is 
$\phi_2(t, i) = \proj{(1,0)}{t, i} = (0, i)$ (see Fig.~\ref{fig:ex-2d-cdag-proj}).

\subsubsection{Summing projections}
\label{sssec:sum-trick}

In some cases, the parts of the \inset of a vertex set associated with two given path relations are actually
disjoint. Let $Q_1$ and $Q_2$ be two such paths, such that $R_{Q_1}^{-1}(P) \cap R_{Q_2}^{-1}(P) = \emptyset$
for any $P \subseteq V\setminus \msources{V}$.
If these are two broadcast paths, then since $R_{Q_i}^{-1}(P)  \subset \In{P}$, any $\K$-bounded set
$P$ satisfies the stronger inequality:
\[\card{\phi_{Q_1}(\rho(P))} + \card{\phi_{Q_2}(\rho(P))} \le \K \]
The same holds if $Q_1$ is a chain circuit and $R_{Q_1}^{-1}(P) \cap R_{Q_2}^{-1}(P) = \emptyset$, by a similar argument.

We say two paths $Q_1$ and $Q_2$ are \emph{independent} for domain $D_S$ if $R_{Q_1}^{-1}(D_S) \cap R_{Q_2}^{-1}(D_S) = \emptyset$.
We can build the \emph{\dfg-path interference graph}: vertices are paths $Q_1,\dots,Q_m$ and there is an edge
between any independent pair of paths.
In this graph, if vertices $Q_{i_1},\dots,Q_{i_t}$ form a clique, then 
\[\text{Any  $\K$-bounded set $P \subseteq D_k$ satisfies} \sum_s\card{\phi_{i_s}(P)} \leq \K.\]

\paragraph{Example} Looking again at Example~1, it is straightforward to check that paths $p_1$ and $p_2$ are independent,
so a $\K$-bounded set $P$ actually satisfies $\card{\phi_1(P)} + \card{\phi_2(P)} \le \K$.

Combining several such inequalities, such that every projection occurs at least once, leads to a general constraint of the form:
\[\sum_{j=1}^m\beta_j \card{\phi_j(E)} \leq \K,\]
for some positive coefficients $\beta_j$.
This is achieved by finding a set of maximal cliques covering all vertices in interference graph,
and summing the corresponding inequalities.
Computing these $\beta_j$'s is the role of function \coeffInterf{} in Algorithm~\ref{alg:path2Q}.

In this case, a tighter bound can be derived: 
\[\card{E} \le \prod_{i=1}^m\card{\phi_j(E)}^{s_j} \le\left(\frac{\K}{\sum_j s_j}\right)^{\sum_j s_j} \prod_{j=1}^{m}\left(\frac{s_j}{\beta_j}\right)^{s_j}.\]

The following lemma establishes this result\xspace.
\begin{lemma}
	\label{lemma:optim}
	Let $0 \leq s_1, s_2, \dots, s_m \leq 1$ and $C > 0$.
	Let $x_j$ be nonnegative integers and $\beta_j > 0$ for $1 \leq j \leq m$ such that
		$\sum_{j=1}^m \beta_j x_j \leq C$.
	Then
	\begin{equation}
		\prod_{j=1}^m x_j^{s_j} \leq   \left(\frac{C}{\sum_j s_j}\right)^{\sum_j s_j} \prod_{j=1}^{m}\left(\frac{s_j}{\beta_j}\right)^{s_j} .
	\end{equation}
\end{lemma}

\begin{proof}
	We use Lagrange multipliers to find the constrained maximum of the function $\psi: x \in \R^m \mapsto \prod_{j=1}^m x_j^{s_j}$.
	\[ L(x, \lambda) = \prod_{j=1}^m x_j^{s_j} - \lambda \left( \sum_{j=1}^m \beta_j x_j - C \right) \]
	Partial derivatives are:
	\begin{align*}
		\ddx{L}{x_j} &= s_j x_j^{s_j-1} \prod_{k \ne j} x_k^{s_k} - \lambda \beta_j \ \   , 1 \leq j \leq m \\
		\ddx{L}{\lambda} &= C - \sum_{j=1}^m \beta_j x_j
	\end{align*}
	Setting them to be 0, we get:
	\[
		\lambda \beta_j x_j = s_j \prod_{k=1}^m x_k^{s_k}, 1 \leq j \leq m
  \]
	Summing for $1 \leq j \leq m$ gives:
	\[ \lambda \sum_{j=1}^m \beta_j x_j = \lambda C = \left( \prod_{k=1}^m x_k^{s_k} \right) \cdot \sum_{j=1}^m s_j \]
	From which we derive:
	\[ x_j = \frac{C s_j}{\beta_j\sum_i s_j} , 1 \leq j \leq m \]
	And finally:
	\[\psi(x) \leq \prod_{j=1}^m x_j^{s_j} = \prod_{j=1}^m \left(\frac{Cs_j}{\beta_j \sum_i s_i}\right)^{s_j} = \left(\frac{C}{\sum_j s_j}\right)^{\sum_j s_j} \prod_{j=1}^{m}\left(\frac{s_j}{\beta_j}\right)^{s_j}. \]

\end{proof}

With this more general formulation, the choice of the $s_j$ coefficients is more involved
than the linear optimization problem of Sec.\ref{ssec:projection}, and is developed in Sec.~\ref{ssec:lb}.

\subsubsection{Kernel subgroup lattice}
As already mentioned, the subgroup lattice (Def.~\ref{def:lattice} used in Lemma~\ref{lm:lattice}) is not necessarily finite, so it is better to build it step-by-step, updating it each time we add a new path. 
We set a time limit for the computation to converge, and do not add the path if this limit is reached. 
Function \subspaceclosure{} in Algorithm~\ref{algo:lattice} tentatively updates the current subgroup lattice with a new one, returning the original lattice in case of a timeout.

\begin{algorithm}[h]
	\small

\Fn{\subspaceclosure}{
  \Input{Lattice of subgroups $\L$, subgroup to add $K$}
  \Output{updated set of subspaces $\L'$}
  $\L'=\L$\;
  \While{\textbf{not} timeout}{
    \lIf{$\exists H\in \L',\ H\cap K\not\in \L'$}{$\L'=\L' \cup \{H\cap K\}$}
    \lElseIf{$\exists H\in \L',\ H+K\not\in \L'$}{$\L'=\L' \cup \{H+K\}$}
    \lElse{\Return $\L'$}
  }                  
  \Return $\L$ if timeout
}

	\caption{\label{algo:lattice}Update the subgroup lattice with a new kernel}
\end{algorithm}

\subsection{Finding paths}
\label{ssec:paths}
The function that generates the set of paths $\mathcal{P} = \{Q_1, \dots, Q_m\}$ for a \dfg-vertex $S$ is named \genpaths\xspace( Alg.~\ref{alg:genpaths}).
Starting from $S$, it uses a simple backward traversal (backward DFS) that favors walking through predecessors with largest domain.
As the number of paths can be combinatorial, \tool sets a timeout to avoid a computational blow-up.
For a path $P = (S_1,\dots,S_t = S)$, we store sub-path relations $R_{S_i\rightarrow S}$ 
for every intermediate statement $S_i$.
This is necessary to ``remember'' exactly which CDAG vertices are included in the computation.

\begin{algorithm}[h]
	\small

\Fn{\genpaths}{
  \Input{a \DFG $G=(\V,\E)$, a statement $S\in \E$}
  \Output{set of paths $\mathcal{P}$}
  starts from $S$ and backward traverse to build any possible path that reach $S$\;
	drop paths for which $R_{S'\rightarrow S}(D_{S'})$ has lower dimensionality than $D_S$\;
	drop paths if $\neg (\isbroadcast(P) \vee \iscircuit(P))$\;
}

	\caption{\label{alg:genpaths}Generate paths}
\end{algorithm}

\subsection{Computing the lower bound}
\label{ssec:lb}
Once we have found a path combination, it is quite straightforward to apply the theoretical
results introduced above.

This is detailed in function \pathtoQ{} in Alg.~\ref{alg:path2Q}.
Here, the role of the function call to \coeffInterf{} is to compute the coefficients $\beta_j$.
It does so by finding a clique cover of the DFG-path independence graph and summing the constraint formed by each clique as explained in Sec.~\ref{sssec:sum-trick}. 
The values for coefficients $s_j$ that satisfy inequalities~\ref{eq:bl1b} are then determined using convex optimization so as to minimize as much as possible the quantity
\begin{equation}
	\label{eq:U}
	U = \left(\frac{\ST}{\sum_j s_j}\right)^{\sum_j s_j} \prod_{j=1}^{m}\left(\frac{s_j}{\beta_j}\right)^{s_j}.
\end{equation}
Indeed this expression being an upper bound on $\card{\rho(P)}$ (see Lemma~\ref{lemma:optim}), minimizing it has the effect of tightening\footnote{Observe that any values of $s_j$ leads to a correct bound} the computed \IO complexity.

The constraints in (\ref{eq:bl1b}) describe a convex polyhedron, but the objective function (\ref{eq:U})
is not convex.
In the basic case when all projections are simply bounded by $\ST$, the objective is 
\[U := \ST^{\sum_j s_j},\]
so a natural objective is to minimize $\sum_j s_j$ in this generalized case.
It can be easily checked that (\ref{eq:U}) is indeed equal to this when $\beta_j = \frac{1}{m}$ for all $j$ and
$s_1 = \dots = s_j$.

Notice that the first factor in the expression of $U$ depends only on the value of the sum. 
So once $\sum s_j$ is fixed, it is natural to minimize the second factor $\prod_{j=1}^{m}\left(\frac{s_j}{\beta_j}\right)^{s_j}$, which is convex as a function of $(s_1, \dots, s_m)$. 
This convex optimization problem can be solved with an appropriate tool, such as IPOPT~\cite{ipopt}. 
he last step amounts to set an appropriate value for $T$ that provides a lower bound of $\left\lfloor\frac{\card{D_S}}{U}\right\rfloor\times \T -|I|$ (see Lemma~\ref{lemma:2S}) as big/tight as possible. 
Here, $D_S$ corresponds to $V \setminus \msources{V}$ in the CDAG view, and $I$ is the frontier of the domain, corresponding to $\msources{V}$. 
$\T$ is chosen as $\frac{1}{{\sum_j{s_j}} - 1} \S$, because it maximizes the first term asymptotically.

Finally, we store the may-spill set corresponding to the CDAG for which this lower bound is valid.

Taking the example given in Figure~\ref{fig:ex-2d}, it is sufficient to check condition (\ref{eq:bl1b}) on $H_1 = \{(0, i)\}, H_2 = \{(t, 0)\}$,
and the optimization problem is:
\begin{align*}
	\text{Minimize\ \ \ \ \ } &s_1 + s_2 \text{ and then } s_1^{s_1} s_2^{s_2}  \\
	\text{s.t.\ \ \ } & s_1 \ge 1,\ \ s_2 \ge 1
\end{align*}
The solution is $s_1 = s_2 = 1$,
and $U = (\ST/2)^2 = \S^2$ (because $\T = \frac{1}{s_1 + s_2 - 1} \S = 1 \cdot \S$).
$\card{D_S} = MN$ and $\card{I} = N+M$, so 
\[\Q \ge \left\lfloor\frac{MN}{\S^2}\right\rfloor\times \S-N - M. \]

\begin{algorithm}[h]
	\small

\Fn{\pathtoQ}{
  \Input{paths $\P = \{P_1,\dots,P_m\}$ with domain $D$ and lattice $\L$ }
  \Output{complexity $\Q$}
	$I:=\bigcup_{P_i\in \P} R_{P_i}^{-1}(D)$\;
  $d := \dim{D}$\;
	$(\beta_1,\dots,\beta_m) := \coeffInterf(\P, D)$\;
  $(s_1, \dots, s_m)$ := convex-opt \{ \\ \Indp
  variables: $\{s_1, \dots, s_m \in \mathbb{Q}^+\}$\\
	objective: minimize $\sum_{j}{s_j}$ and then $\prod_{j}\left(\frac{s_j}{\beta_j}\right)^{s_j}$\\
  constraints:
	$\forall H \in \L, \ \ \sum_{j} s_j \rk{\phi_j(H)} \geq \rk{H}$ where $\phi_j = \text{proj}^{\perp}_{\kernel(P_j)}$\} \;
  \Indm
	$\T$ := $\frac{1}{{\sum_j{s_j}} - 1} \S$ \;
	$U$ := $\prod_{j=1}^m \left(\frac{\ST s_j}{\beta_j \sum_i s_i}\right)^{s_j}$\;
  $\Q:=\max\left(\left\lfloor\frac{|D|}{U}\right\rfloor\times \T-|I|,0\right)$\;
	$\Q.\mayspill := \mayspill(\P, D)$\;
}
\Fn{\coeffInterf}{
	\Input{paths $\P = \{P_1,\dots,P_m\}$, domain $D$}
	\Output{coefficients $(\beta_1,\dots,\beta_m)$ such that $\sum{\beta_j \phi_j(E)} \leq \K$ for any $\K$-bounded set $E$}
	$G$ := graph with $V := P_1,\dots, P_m$ and $E := (P_i, P_j),\  R_{P_i}^{-1}(D)\cap R_{P_i}^{-1}(D) \ne \emptyset$\;
	$\mathcal{I}$ := set of maximal independent sets of $G$ such that every node belongs to at least one set (greedy construction)\;
	$\beta_j := \#\{I \in \mathcal{I}, P_j \in I\} / \card{\mathcal{I}}$\;
}

\Fn{\mayspill}{
	\Input{paths $\P = \{P_1,\dots,P_m\}$, domain $D$}
	\Output{may-spill set $D^{\text{ms}}$}
	$D^{\text{ms}} := \emptyset$\;
	\lForEach{broadcast path $P_i = (S_0, S_1,\dots,S_t = S) \in \P$}{
		$D^{\text{ms}} := D^{\text{ms}} \cup \left( \bigcup_{j = 0}^t R_{S_j \rightarrow S}^{-1}(D) \right)$ }
	\lForEach{chain circuit $P_i = (S_0, S_1,\dots,S_t = S) \in \P$}{
		$D^{\text{ms}} := D^{\text{ms}} \cup \left( \bigcup_{j = 1}^t R_{S_j \rightarrow S}^{-1}(D) \right)
		\cup \left(R_{S_0 \rightarrow S}^{-1}(D) \cap \left( \bigcup_{k \ne j} R_{P_k}^{-1}(D)\right) \right)$ }
}

\Fn{\kernel}{
	\Input{path $P_j$}
	\Output{linear space $K$ such that the orthogonal projection $\phi_j = \text{proj}^{\perp}_{\kernel(P)}$}
	\lIf{\isbroadcast($P_j$)}{
		\Return $\kernel (j_1 \dots j_{d'})$ where $R_{P_j} = \{T[j_1,\dots,j_{d'}] \rightarrow [i_1,\dots,i_d]:\dots \}$
		}
	\lIf{\iscircuit($P_j$)}{
		\Return $(\delta_1,\dots,\delta_d)$ where $R_{P_j} = \{S[i_1,\dots,i_d] \rightarrow S[i_1+\delta_1,\dots,i_d+\delta_d]: \dots \}$
		}
}

        \vspace{\baselineskip}
	\caption{\label{alg:path2Q}Derivation of a lower bound from a path combination with the partition method}
\end{algorithm}

\section{Wavefront bound derivation}
\label{sec:wf}

\subsection{Theoretical results}
An alternative way to derive data movement lower bounds in the no-recomputation model is the \emph{wavefront} abstraction. 
At any point in an execution of a RW-game, the wavefront is the set of vertices that have been computed but whose result is still needed by some successor (sometimes called the set of \emph{live} vertices). 
If the size of the wavefront at some point in the execution is greater than the size of the fast memory, then necessarily some vertices have to be spilled to the slow memory and thus loaded  using rule $(R1)$. 
This is formalized in the definition an lemma below:

\begin{definition}[Wavefront]
	\label{def:wavefront}
	Let $\RR$ be an execution of the RW-game on a CDAG $G$, and $v$ a vertex of $G$.
	Consider the time $t$ in the execution 
	just before $v$ has been computed (i.e. just after a white pebble has been placed on $v$ using rule $(R2)$).
	The \emph{wavefront} $W_{\RR}(v)$ in execution $\RR$ is the set of vertices that, a time $t$,
	have been computed (i.e. have a white pebble) but have some successor that does not.
\end{definition}

 \begin{lemma}[Min-wavefront~\cite{elango-spaa2014}]
   \label{lemma:wf}
	 Let $\S$ be the capacity of the fast memory, and $G = (V, E)$ be a CDAG.
	 Let $w_G^{\textrm{min}} = \min_{\RR} \left( \max_{v \in V} \card{W_{\RR}(v)} \right)$, so that any valid RW-game on $G$ has a wavefront of size at least $w_G^{\textrm{min}}$.
	 Then, \[\Q\ge w_G^{\textrm{min}}-\S.\]
 \end{lemma}

 This lemma is quite general and cannot be applied directly, as $w_G^{\textrm{min}}$ is not usually computable. 
 We thus provide the following result, which uses a condition on the structure of the CDAG to get a lower bound on the size of a wavefront in any execution.

\begin{corollary}
 \label{cor:wf}
	Let $G = (V, E)$ be a CDAG, and $V_1,V_2$ be disjoint subsets of $V$ such that every vertex in
	$V_2$ is reachable from every vertex in $V_1$ through some path in $G$.
 Let $L_1,\dots,L_m$ be disjoint paths in $G$, starting in $V_1$ and ending in $V_2$.
 More formally:
 \[\left\lbrace
			 \begin{array}{ll}
				 \forall L_j = (v_1^j,\dots,v_{l_j}^j), & v_1^j \in V_1 \text{\ and\ } v_{l_j}^j \in V_2 \\
	  V_1\times V_2 \subset E^*
			 \end{array}\right.
 \]

 Then,
	\[w_G^{\textrm{min}} \ge m. \]
	By Lemma~\ref{lemma:wf}, this implies:
	\[\Q \geq m - \S. \]
\end{corollary}

\begin{proof}
	Let $v$ be the first vertex to be computed among vertices of $V_2$ in some fixed RW-game on $G$.
	By the second condition, every vertex in $V_1$ must have been computed since $v$ depends on all of them.
	Just before a red pebble is placed on $v$, no vertex in $V_2$ has a red pebble.
	Therefore there is a \alive vertex (that has a red pebble and a successor without one) in every path $L_j$.
	Thus $w_G^{\textrm{min}} \ge m$.
\end{proof}

The common case to use this technique to get a strong data movement lower bound is to combine it with the parametric CDAG decomposition (Sec.~\ref{ssec:loop}).

\paragraph{Example}In the example of Fig.~\ref{fig:re1d}, Corollary~\ref{cor:wf} can be applied on each subgraph, with $V_1$ and $V_2$ chosen as shown on Fig.~\ref{fig:re1d-wf}. It can be easily checked that every vertex in $V_1$ can reach every vertex in $V_2$, and there are $N$ disjoint paths $L_j$, so this gives a lower bound
\[ \Q(G_{|V_t}) \ge N - \S.\]

\begin{figure}
	 \includegraphics[width=0.4\textwidth]{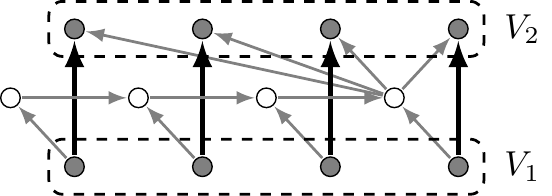}
	 \caption{\label{fig:re1d-wf}Application of Cor.~\ref{cor:wf} on a sub-CDAG. Paths $L_j$ (single edge) are shown in bold. }
	 
\end{figure}

\subsection{Implementation}
To apply this technique, our algorithm tries to uncover a set of disjoint paths satisfying the hypotheses of Corollary~\ref{cor:wf}. 
To reduce the search space, Algorithm~\ref{alg:wf2Q} looks for a much more constrained pattern, in which all disjoint paths $L_j$ begin and end in different instances of some statement $S$, with an increment in the innermost parametrized loop index (see Sec.~\ref{ssec:loop}). 
This amounts to finding an injective circuit in the \dfg with a relation of the form $\{S[I_1 \dots I_d,i_{d+1}\dots i_D] \rightarrow S[I_1 \dots I_d+1,i_{d+1}\dots i_D]\}$.

Intuitively, we look at two ``slices'' of the CDAG, each of them of dimension $D-d$, representing two successive iterations of the body of the loop iterating over dimension $d$ (with index $I_d$ and $I_d+1$, respectively). 
We try to find subsets of these two slices (corresponding to $V_1$ and $V_2$) such that every vertex in the first one can reach every vertex in the second one.

In Algorithm~\ref{alg:wf2Q}, the first loop computes the following relations:
\begin{itemize}
	\item $R_{S\rightarrow S}$ is the union of all path relation of elementary circuits from $S$ to itself,
	\item $\RL_{S\rightarrow S}$ is the union of all \emph{affine} path relations (where every subpath relation is affine) of elementary circuits $S \rightarrow S$,
  \item $R_{S\rightarrow *}$ is the union of all path relations of elementary paths from $S$ to any other \dfg-vertex.
\end{itemize}

This is then used to compute the relation $R_{I_d}$.
Here $I_d$ is the index of the innermost parametrized loop, and $R_{I_d}$
is the restriction of $\RL_{S\rightarrow S}$ to paths that do exactly one step along dimension $d$, not changing any other index.
This represents a set of disjoint paths going from instances of statement $S$ in ``slice'' $I_d$ to instances of $S$ in ``slice'' $I_d + 1$.

To apply Corollary~\ref{cor:wf}, we need to restrict $R_{I_d}$ to a domain $W$ where every starting vertex reaches every ending vertex.
To do this, we first compute $X =\left(R_{\textrm{complete}}-\left(R_{S\rightarrow S}\right)^*\right)\left(\domain{R_{I_d}}\right)$,
that is, all vertices in ``slice'' $I_d + 1$ that are \emph{not reachable} from ``slice'' $I_d$.
We then take $W = \domain{R_{I_d}} - R_{I_d}^{-1}\left(X\right)$, that is, we only keep vertices in ``slice'' $I_d$ that can reach every vertex in slice $I_d + 1$.
Application of Corollary~\ref{cor:wf} with $V_1 = W$, $V_2 = R_{I_d}(W)$ and $\{L_j\}_j = R_{I_d}$ gives $\Q \ge \card{W} - \S$.

\begin{algorithm}[h]
	\small

\Fn{\wavefronttoQ}{
  \Input{\dfg $G = (\V, \E)$, vertex $S\in \V$, parametrized dimensions $\Omega_d$}
  \Output{lower bound $\Q$}
  \BlankLine
  $D:=\dim{S}$\;
  $R_{S\rightarrow S}:=\emptyset$; $\RL_{S\rightarrow S}:=\emptyset$; $R_{S\rightarrow *}:=\emptyset$\;
  \ForEach{$S_j\in \V$ in topological order, from $S$ (excluded) to $S$ (included)}{
    $A:=\{(S_i,S_j)\in E,\ R_{S_i\rightarrow S_j}\textrm{ is affine and }R_{S_i\rightarrow S_j}^{-1}\textrm{ injective}\}$\;
    $\RL_{S\rightarrow S_j}:=\bigcup_{(S_i,S_j)\in A} \RL_{S\rightarrow S_i}\circ R_{S_i\rightarrow S_j}$\;
    $R_{S\rightarrow S_j}:=\bigcup_{(S_i,S_j)\in E} R_{S\rightarrow S_i}\circ R_{S_i\rightarrow S_j}$\;
    $R_{S\rightarrow *}:=R_{S\rightarrow *} \cup R_{S\rightarrow S_j}$\; 
  }
  $R_{I_d}$ := $\RL_{S\rightarrow S}\cap \{S[I_1 \dots I_d,i_{d+1}\dots i_D] \rightarrow S[I_1 \dots I_d+1,i_{d+1}\dots i_D]\}$\;
  $R_{\textrm{complete}}:=\{S[I_1 \dots I_d,i_{d+1}\dots i_D] \rightarrow S[I_1 \dots I_d+1,i'_{d+1}\dots i'_D]\}$\;
	$W:=\domain{R_{I_d}}-R_{I_d}^{-1}\left(\left(R_{\textrm{complete}}-\left(R_{S\rightarrow S}\right)^*\right)(\domain{R_{I_d}})\right)$\; 
  $\Q := \max\left(|W|-\S,0\right)$\;
  $\Q.\mayspill:=W\cup \left(R_{S\rightarrow *}(W)\cap R^{-1}_{S\rightarrow *}\left(R_{I_d}(W)\right)\right)$
}

        \vspace{\baselineskip}
	\caption{\label{alg:wf2Q}Derivation of a lower bound with the wavefront method}
\end{algorithm}

\section{Complete framework}
\label{sec:complete}

\subsection{\dfg construction}
Our front end (PET~\cite{verdoolaege2012pet}) takes as input a program in C where the to-be analyzed regions (SCoPs -- Static Control Parts) are delimited by \verb+#pragma scop+ and \verb+#pragma endscop+ annotations. 
For PET, all array accesses are supposed not to alias with one another. 
Any scalar data is assumed to be atomic and all of the same size: 
our CDAG is not weighted (which is a limitation of our implementation and not a conceptual limitation of the approach). 
As illustrated by the example of Fig.~\ref{fig:ex-2d} and~\ref{fig:ex-2d-2} (multidimensional-)array accesses are affine expressions of static parameters and loop indices. 
A static parameter can be the result of any complex calculation but has to be a fixed value for the entire execution of the region. 
Loop bounds and more generally control tests follow the same rules (affine expressions). 
As a consequence, the iteration space is a union of (parametric) polyhedra, and memory accesses (read and writes) are piecewise affine functions. 
This representation of the region execution that fits into the polyhedral framework~\cite{polyhedron} allows to compute data dependencies using standard data-flow analyses.

PET outputs a polyhedral representation of the input C
program, from which we extract a \emph{\DFG (\dfg)} $G = (\V, \E)$ (see Sec.~\ref{sec:DFG}).

\subsection{Instances of parameter values}
As briefly explained in Sec.~\ref{sec:decomp}, to generate bounds that are as tight as possible, our heuristic needs to take decisions.
Such decisions are based on our ability to compare the size of two different domains sizes or even the complexity of two different sub-CDAGs.
The overall framework being parametric (it provides complexities that are functions of parameter values and cache size), a total order is obtained by considering a specific instance of parameter values, taken as an additional input alongside the C program.
One needs to outline that a specific instance of parameter values is \emph{not} considered by the algorithm as a precondition:
For a given instance, the computed lower bound expression is universal i.e. is correct for \emph{any} parameter values.
For completeness, several instances are considered, and to each instance $\i$ is associated a complexity $\Q^\i$.
As we have $\Q\ge\Q^\i$ for any instance, denoting $\I$ the set of all considered instances, they are simply combined as:
\[\Q=\max_{\i\in\I}\left(\Q^\i\right).\]

\subsection{Main algorithm}
\label{ssec:mainloop}
Alg.~\ref{alg:mainloop} contains the skeleton of the main part of \tool, with links to corresponding subsections.
Its outermost loop (Line~\ref{l:main:param}) corresponds to the loop parametrization detailed in Sec.~\ref{ssec:loop}: for each loop depth $d$, outermost indices are fixed (as parameter $\Omega_d$ -- Line~\ref{l:main:omega}), and parametrically computed lower bounds are summed (when not interfering -- Line~\ref{l:main:inter}) over all iterations (Line~\ref{l:main:sumloop} in \addparam).
The loop on statements $S$ (Line~\ref{l:main:S}) allows to decompose the full CDAG into as many ``$S$-centric'' sub-CDAGs.
The so-obtained bounded set of lower bounds $\Ginterf$ are combined using procedure \maxcover (Line~\ref{l:main:maxcover}) as described in Sec.~\ref{ssec:maxcover}. To take compulsory misses into account, the size of the input data of the program is added to the expression.

For each statement $S$, both techniques ($K$-partition and wavefront resp. Line~\ref{l:main:partition} and Line~\ref{l:main:wavefront}) generate lower bounds.
As opposed to the implicitly considered ``$S$-centric'' sub-CDAGs for the wavefront reasoning, an ``$S$-centric'' sub-CDAGs for the $K$-partition reasoning (which is built by finding a set $\P$ of \dfg-paths that terminate at $S$ -- Lines~\ref{l:main:gpathb}-\ref{l:main:gpathe} through function \genpaths) does not necessarily spans all the $S$-vertices ($D_S$) of the CDAG.
So several (non-intersecting) sub-CDAGs can be built until no more interesting lower bound can be derived (Line~\ref{l:main:gpathe}). 
Hence, for each statement $S$, a copy $G'$ of the CDAG $G$ is made: as new set of paths ($S$-centric sub-CDAGs) and corresponding lower bounds are computed, the corresponding may-spill set is removed from $G'$ (Line~\ref{l:main:prunespill}).

\begin{algorithm}[h]
	\small

\Fn{\mainloop}{
  \Input{\DFG $G=(\V,\E)$, an instance $I$}
  \Output{lower bound $\Qlow$}
	$\Ginterf = \emptyset$\;
	\tikzmark{l1a}Let $D$ be the max loop depth\;
  \ForEach{loop level $0\leq d<D$}{\label{l:main:param}
    \ForEach{statement $S\in \V$ surrounded by at least $d+1$ loops}{\label{l:main:S} 
			 $\Omega_d:=[I_1,\dots,I_d]\rightarrow\{S[i_1,\dots,i_D]:\ i_1=I_1\wedge \dots i_d=I_d\}$\tikzmark{l1b}\;\label{l:main:omega}
        Let $G'$ be a copy of $G$\;
				\tikzmark{pa}\While{elapsedTime < timeout}{
          Let $D_S$ be the parametrized domain of $S$ in $G'$\;
         $\P:=\emptyset, \L:=\emptyset$\;\label{l:main:gpathb} 
					\ForEach{$P_i\in \genpaths(G',S,\Omega_d)$ (in increasing order of $\dim{\kernel(P_i)})$\tikzmark{pb}}{
						\If{$\card{D_S \cap \domain{P_i}} \ge \gamma\card{D_S}$}{
               $K_i:=\kernel(P_i)$\;
               \If{$\L:=\subspaceclosure(\B,K_i)$ changed}{
								 $D_S:=D_S \cap \domain{P_i}$\;
                 $\P:=\P\cup P_i$\;
               }
            }
          }
           \lIf{$\P=\emptyset$}{\textbf{exit} while loop}\label{l:main:gpathe}
				$(\Ginterf, G')=\addparam(\Ginterf, G', \pathtoQ(\mathcal{P},D_S,\L,\Omega_d))$\tikzmark{pc}\;\label{l:main:partition}
				}
				\tikzmark{wa}$(\Ginterf, G')=\addparam(\Ginterf, G', \wavefronttoQ(S,\Omega_d))$\tikzmark{wb}\;\label{l:main:wavefront}
     }
  }
	\tikzmark{sa}$\Qlow=\textrm{input\_size}(G) + \max(0, \maxcover(\Ginterf))$\tikzmark{sb}\;\label{l:main:maxcover}
}

\Fn{\addparam}{
	\Input{set of global bounds $\Ginterf$, \dfg $G'$, parametrized bound $Q(\Omega)$}
	\Output{updated $\Ginterf$, $G'$}
	\tikzmark{l2a}
	\uIf{$\left[\Omega\neq\Omega' \Rightarrow Q.\interf(\Omega)\cap Q.\interf(\Omega')=\emptyset\right]$\tikzmark{l2b}}{\label{l:main:inter}
		$Q:=\sum_{\Omega} Q(\Omega)$\;\label{l:main:sumloop}
		$Q.\mayspill:=\bigcup_{\Omega} Q.\mayspill(\Omega)$\tikzmark{l2c}\;
		$\Ginterf = \Ginterf \cup \{Q\}$\;
    $G':=G' \setminus Q.\mayspill$\;\label{l:main:prunespill}
	}
}
\begin{tikzpicture}[overlay, remember picture,
	rel/.style={fill, draw=orange!50!black, fill=orange!20, opacity=0.3},
	rep/.style={fill, draw=blue, fill=blue!20, opacity=0.3},
	al/.style={->, thick, orange,shorten > = 4pt, shorten < = 3pt}]
	\node [above left = 5pt and 0pt of pic cs:l1a] (L1a) {};
	\node [below right = 0pt and 5pt of pic cs:l1b] (L1b) {};
	\node [above left = 5pt and 0pt of pic cs:l2a] (L2a) {};
	\node (l2b) at (pic cs:l2b) {};
	\node (l2c) at (pic cs:l2c) {};
	\node (l2d) at (l2b |- l2c) {};
	\node [below right = -3pt and 20pt of l2d] (L2b) {};
	\draw[rel] (L1a) rectangle (L1b);
	\draw[rel] (L2a) rectangle (L2b);

	\node [above right = .3cm and 1cm of L1b, orange!50!black] (Ls) {\Large Sec.~\ref{ssec:loop}};
	\draw [al] ($(L1b) + (0, 10pt)$) -- (Ls);
	\draw [al] ($(L2b) + (0, 10pt)$) -- (Ls);

	\node [above left = 5pt and 0pt of pic cs:pa] (Pa) {};
	\node (pb) at (pic cs:pb) {};
	\node (pc) at (pic cs:pc) {};
	\node (pd) at (pb |- pc) {};
	\node [below right = -4pt and 10pt of pd] (Pb) {};
	\draw[rep] (Pa) rectangle (Pb);
	\node [above right = 1.2cm and .4cm of Pb, blue] (Ps) {\Large Sec.~\ref{sec:partition}};

	\draw[rel, green!50!black, fill=green!20] ($(pic cs:wa) + (-2pt, 8pt)$) rectangle ($(pic cs:wb) + (3pt, -3pt)$);
	\node [right = 0.4cm of pic cs:wb, green!50!black] (Ws) {\Large Sec.~\ref{sec:wf}};

	\draw[rel, red!50!black, fill=red!20] ($(pic cs:sa) + (-2pt, 8pt)$) rectangle ($(pic cs:sb) + (3pt, -3pt)$);
	\node [right = 0.4cm of pic cs:sb, red!50!black] (Ss) {\Large Sec.~\ref{ssec:maxcover}};
\end{tikzpicture}

	\caption{\label{alg:mainloop} Main procedure that computes $\Qlow$ for the program by combining lower bound of sub-CDAGs obtained through $K$-partition or wavefront reasoning}
\end{algorithm}

\section{Experimental Evaluation}
\label{sec:exp}
\tool was implemented in C, using ISL-0.13~\cite{isl-manual}, 
barvinok-0.37~\cite{barvinok1994polynomial} and PET-0.05~\cite{verdoolaege2012pet}.
We also used GiNaC-1.7.4~\cite{bauer2002ginac} for the manipulation of symbolic expressions, and PIP-1.4.0~\cite{Fea88}
for linear programs.
\tool takes as input an affine C program and outputs a symbolic expression
for a lower bound on \IO complexity as a function of the problem size parameters of the
program and capacity of fast memory.

\tool was applied to all programs in the \polybench/C-4.2.1 benchmark suite~\cite{polybench}. To evaluate the quality of the results produced by \tool, we manually generate tiled versions of each kernel, then manually compute parametric data-movement costs as a function of tile sizes and cache size, then manually find the optimal tile sizes and thereby, finally, derive a manually optimized data-movement cost for this kernel. By forming the ratio of the total number of operations and the data-movement cost, we then generate $\OIopt$. In this derivation, we assume that we have explicit control of the cache. Then $\OIopt$ is compared with an operational intensity upper-bound obtained by forming the ratio of the number of operations and the data movement lower bound generated by \tool: $\OIup$. (We always must have $\OIup\geq\OIopt$.)

Let us use \texttt{jacobi-1d} as an example to illustrate all this.
\tool computes a lower bound expression $\Qlow$ on the number of loads needed for any schedule of the \texttt{jacobi-1d} kernel:
\[\Qlow  = 2+ N + \max\left(0, \frac{T N}{4\S} - N - T - \frac{1}{4}\frac{N}{\S} - \frac{3}{4} \frac{T}{\S} - \S + 5\right).\]
The first term is the input data size, and the second term is obtained by the partitioning technique.
Since the expression of $\Qlow$ can be quite large, we automatically simplify to $\Qinfty$ by only retaining the
asymptotically dominant terms, assuming all parameters $N,M\dots$ and cache size $\S$ tend to infinity, and $\S = o(N,M,\dots)$,
\[\Qinfty = \frac{T N}{4 \S} .\]
Finally, from $\Qinfty$ and the fact that the \texttt{jacobi-1d} kernel performs $6TN$ operations, we compute an upper bound for the \OI of any schedule of the \texttt{jacobi-1d} kernel,
		\[\OIup = \frac{6 T N}{\Qinfty} = 24 \S.\]

Our manually generated schedule uses a horizontal band of width $\S/2$. It has an \IO
of $\Qopt^{\infty}= 4NT/\S$, leading to $\OIopt = \frac{6 T N}{\Qopt^{\infty}} = \frac{3 \S}{2}$

In this case, $\OIup$ and $\OIopt$ are not equal. Such a gap means that it is possible to (1) either increase the data-movement lower bound ($\OIup$) generated by \tool, (2) or find a better schedule for this kernel with less data transfer so as to decrease $\OIopt$; (3) or both improvements are possible.
$\OIup$ is not necessarily tight. $\OIopt$ is not necessarily the highest achievable \OI. We only can conclude so if both quantities are equal.

\renewcommand{\arraystretch}{1.2}
\begin{table}[h]
\footnotesize
   \caption{Operational intensity upper and lower bounds for \polybench Benchmarks}

\begin{tabular}{|l||c|c|c||c|c|c|}
\hline
\textbf{kernel}           & \textbf{\# input data}    & \textbf{\#ops}       & ratio            & \textbf{$\OIup$}                        & \textbf{$\OIopt$}               & ratio                           \\ 
\hline                         
  \textsf{2mm}            & $ \phantom{+} N_i N_k + N_k N_j $ & $\phantom{+} NiN_jN_k$ & -- & $\sqrt{S}$               & $\sqrt{S}$ & $1$ $\checkmark$ \\
                          & $          +  N_j N_l + N_i N_l $ & $ + N_iN_jN_l$         &    &                          &            &                  \\
  \textsf{3mm}            & $ \phantom{+} N_i N_k + N_k N_j $ & $\phantom{+} NiN_jN_k + N_jN_lN_m$ & -- & $\sqrt{S}$               & $\sqrt{S}$ & $1$ $\checkmark$ \\
                          & $          +  N_j N_m + N_m N_l $ & $ + N_iN_jN_l$                     &    &                          &            &                  \\
  \textsf{cholesky}       & $\frac{1}{2}N^2$          & $\frac{1}{3}N^{3}$     & $\frac{2}{3}N$  & $2 \sqrt{S}$                             & $\sqrt{S}$                      & $2$                           \\
  \textsf{correlation}    & $MN$                      & $M^{2} N$              & $M$              & $2 \sqrt{S}$                            & $\sqrt{S}$                      & $2$                           \\
  \textsf{covariance}     & $MN$                      & $M^{2} N$              & $M$              & $2 \sqrt{S}$                            & $\sqrt{S}$                      & $2$                           \\
  \textsf{doitgen}        & $N_p N_q N_r$             & $2 N_q n_r N_p^{2}$    & $2N_p$           & $\sqrt{S}$                              & $\sqrt{S}$                      & $1$ $\checkmark$              \\
  \textsf{fdtd-2d}        & $3N_xN_y$                 & $11 N_xN_y T$          & $\frac{11}{3} T$ & $22 \sqrt{2} \sqrt{S}$                  & $\frac{11}{24}\sqrt{3}\sqrt{S}$ & $\frac{48\sqrt{2}}{\sqrt{3}}$ \\
  \textsf{floyd-warshall} & $N^2$                     & $2 N^{3}$              & $2N$             & $2 \sqrt{S}$                            & $\sqrt{S}$                      & $2$                           \\
  \textsf{gemm}           & $N_iN_j+N_jN_k+N_iN_k$    & $2 N_iN_jN_k$          & --               & $\sqrt{S}$                              & $\sqrt{S}$                      & $1$ $\checkmark$              \\
  \textsf{heat-3d}        & $N^3$                     & $30 N^{3} T$           & $30T$            & $\frac{160}{3\sqrt[3]{3}} \sqrt[3]{S}$  & $\frac{5}{2} \sqrt[3]{S} $      & $\frac{64}{3\sqrt[3]{3}}$     \\
  \textsf{jacobi-1d}      & $N$                       & $6 N T$                & $6T$             & $24 S$                                  & $\frac{3}{2} S$                 & $16$                          \\
  \textsf{jacobi-2d}      & $N^2$                     & $10 N^2 T $            & $10T$            & $15 \sqrt{3} \sqrt{S}$                  & $\frac{5}{4} \sqrt{S} $         & $12\sqrt{3}$                  \\
  \textsf{lu}             & $N^2$                     & $\frac{2}{3} N^{3}$    & $\frac{2}{3} N$  & $\sqrt{S}$                              & $\sqrt{S} $         & $1$ $\checkmark$                 \\
  \textsf{ludcmp}         & $N^2$                     & $\frac{2}{3} N^{3}$    & $\frac{2}{3} N$  & $\sqrt{S}$                              & $\sqrt{S} $         & $1$ $\checkmark$                 \\
  \textsf{seidel-2d}      & $ N^2 $                   & $ 9 N^{2} T $          & $9T$             & $\frac{27\sqrt{3}}{2}\sqrt{S} $         & $ \frac{9}{4} \sqrt{S} $        & $6\sqrt{3}$                   \\
  \textsf{symm}           & $\frac{1}{2}M^2+2MN$      & $2 M^2 N$              & --               & $\sqrt{S}$                              & $\sqrt{S}$                      & $1$ $\checkmark$              \\
  \textsf{syr2k}          & $\frac{1}{2}N^2+2MN$      & $2 M N^2$              & --               & $2 \sqrt{S}$                            & $\sqrt{S}$                      & $2$                           \\
  \textsf{syrk}           & $\frac{1}{2}N^2+MN$       & $M N^2$                & --               & $2 \sqrt{S}$                            & $\sqrt{S}$                      & $2$                           \\
  \textsf{trmm}           & $\frac{1}{2}M^2+MN$       & $M^2 N$                & --               & $ \sqrt{S} $                            & $\sqrt{S}$                      & $1$ $\checkmark$              \\ 
\hline        
  \textsf{atax}           & $MN$                      & $4 MN$                 & $4$              & $4$                                     & $4$                             & $1$ $\checkmark$             \\
  \textsf{bicg}           & $MN$                      & $4 M N$                & $4$              & $4$                                     & $4$                             & $1$ $\checkmark$              \\
  \textsf{deriche}        & $HW$                      & $32 H W$               & $32$             & $32$                                    & $\frac{16}{3}$                  & $6$                           \\
  \textsf{gemver}         & $N^2$                     & $10 N^{2}$             & $10$             & $10$                                    & $5$                             & $2$                           \\
  \textsf{gesummv}        & $2N^2$                    & $4 N^{2}$              & $2$              & $2$                                     & $2$                             & $1$ $\checkmark$              \\
  \textsf{mvt}            & $N^2$                     & $4N^2$                 & $4$              & $4$                                     & $4$                             & $1$ $\checkmark$              \\
  \textsf{trisolv}        & $\frac{1}{2}N^2$          & $N^2$                  & $2$              & $2$                                     & $2$                             & $1$ $\checkmark$              \\
\hline                                                                   
  \textsf{adi}            & $N^2$                     & $30 N^{2} T$           & $30T$            & $30$                                    & $5$                             & $6$                          \\
  \textsf{durbin}         & $N$                       & $2 N^{2}$              & $2N$             & $4$                                     & $\frac{2}{3}$                   & $6$                           \\
\hline                                                                    
  \textsf{gramschmidt}    & $MN$                      & $2 M N^{2}$            & $2N$             & $2 \sqrt{S}$                            & $1$                             & $2 \sqrt{S}$                  \\
  \textsf{nussinov}       & $\frac{1}{2}N^2$          & $\frac{1}{3}N^{3}$     & $\frac{2}{3}N$   & $2 \sqrt{S}$                            & $1$                             & $2\sqrt{S}$                   \\
\hline
\end{tabular}

	 \label{tab:polybench}
\end{table}

\subsection{Parametric Bounds for \OI}
\label{ssec:expe1}

Table~\ref{tab:polybench} reports, for each kernel in \polybench:
\begin{itemize}
	\item the size of the input data as well as the number of operations\footnote{\# ops are given as an indication, as some benchmarks operate on integers, and some implementations of classical linear algebra primitives in \polybench are disputable.}, and the ratio between them;
	\item the parametric upper bound on operational intensity $\OIup = \frac{\text{\# ops}}{\Qlow}$ from \tool;
	\item  the parametric lower bound $\OIopt = \frac{\# ops}{\Qopt}$ obtained by hand;
	\item the ratio $\frac{\OIup}{\OIopt}$, assessing the tightness of the bounds.
\end{itemize}

The 30 reported benchmarks can be divided into four categories, corresponding to table divisions:
\begin{enumerate}
	\item (19 kernels) The ratio $\frac{\text{\# ops}}{\text{\# input data}}$ is high, so this is an indication for potential tiling. In these cases, we manually find that tiling is actually possible. \tool gives a non-trivial \OI upper bound that is within a constant of the manually obtained \OI lower bound $\OIopt$.
		The bound is asymptotically tight for 8 of them, and within a factor of 2 for an additional 6. Except for matrix multiplication (\texttt{gemm}), where it matches the best published bound, these are all improvements over previously published results.
	\item (7 kernels) The ratio $\frac{\text{\# ops}}{\text{\# input data}}$ is a constant: clearly, these cases do not provide enough operations to enable data reuse. The reported lower bound by \tool is $\text{\# input data}$, which is asymptotically tight for 5 of them, and within a factor of 2 for 1 more.
	\item (2 kernels) The ratio $\frac{\text{\# ops}}{\text{\# input data}}$ is high which does not discard potential for tiling and high \OI. Our best manual schedule leads to a constant \OI which is arbitrarily far from this optimistic ratio. \tool proves that the code is not tileable, the best achievable \OI is a constant. \tool finds this upper bound on \OI thanks to the wavefront technique. This is better by at least a factor of $\sqrt{\S}$ than any bound that could be obtained by geometric reasoning.
	\item (2 kernels) There is an arbitrarily large discrepancy between $\OIup$ and $\OIopt$.
		Visual examination shows that, for these cases, \tool is too optimistic. These codes are actually not tileable in all dimensions, and we believe that it is possible, using more advanced techniques that are currently out of the scope of \tool, to prove a smaller matching \OI upper bound.
\end{enumerate}

\paragraph{Additional remarks}For \texttt{lu} and \texttt{floyd-warshall}, the analysis automatically decomposes the instances of a single statement into appropriate subdomains and accumulates the bounds, leading to a tighter bound than would have been obtained without such a decomposition.

For several stencil-like computations,  the bound generated by \tool is rather loose  ($16\sqrt{6}$ for \texttt{fdtd-2d}!).
This is due to two  limitations of \tool:
1.~the first (which is a theoretical limitation) is our inability to apply the technique of Sec.~\ref{sssec:sum-trick} because of a possible overlap of the ``chain'' dependencies;
2.~the second (which is an implementation limitation) is because, in the presence of several arrays, \tool only selects one per dimension, losing the opportunity to tighten the inequality constraints.

The complete symbolic expressions output by \tool are available in Appendix~\ref{appendix:full-bounds}.

\subsection{Comparison with machine balance for a specific architecture}
\label{ssec:expe2}

In order to illustrate a practical example of use of the lower bounds derived by \tool, we use the PLuTo~\cite{pluto} tiling algorithm to generate a tiled schedule, from which an idealized data-movement cost is determined using a cache simulator (Dinero~\cite{dinero}).
As opposed to $\OIup$ (derived from our tool \tool) and $\OIopt$ (derived from a manually derived schedule using an optimal cache replacement policy) that both provide a parametric operational intensity for each benchmark, the achieved \OI obtained via PLuTo ($\OIpluto$) provides a numerical operational intensity (using the \texttt{LARGE} data set) with a LRU replacement policy.
The architecture we consider in the rest of this section is a single core, with a machine balance of 8 words/cycle and a fast memory capacity of 256 kB. This more or less corresponds to L2/L3 transfers on a last-generation Intel CPU (Skyline-X), with SIMD AVX512 units.

Figure~\ref{fig:balance} reports $\OIup$ (by instantiating the parametric formula) and $\OIpluto$ for each of the \polybench kernels. 
The gaps between $\OIpluto$ and $\OIup$ that can be observed in Figure~\ref{fig:balance} come from different factors: 
1.~$\OIpluto$ uses a cache simulator, while $\OIup$ assumes an explicit (optimal) control of the cache; 
2.~$\OIpluto$ schedule space is limited due to only considering a fixed-size tiling;
3.~There already exists a gap between $\OIup$ and $\OIopt$ as reported in Table~\ref{tab:polybench}. 
For example, for \texttt{gemm}, in Figure~\ref{fig:balance}, we observe a factor of 6 between $\OIpluto$ and $\OIup$ while $\OIopt$ and $\OIup$ match in Table~\ref{tab:polybench}. 
We want to demonstrate that, despite this apparent gap between theory and practice, we can draw many useful practical conclusions from our theoretical tool, \tool.

For example, in order to know whether our application will be compute-limited or bandwidth limited on  the specific architecture, we plotted the machine balance on Figure~\ref{fig:balance}. We can then observe the three different scenarios:
\begin{enumerate}
\item (18 kernels) $\OIpluto$ is above the machine balance. Even in the cases where the upper bound $\OIup$ is quite larger than $\OIpluto$, indicating that data movement of $\OIpluto$ could potentially be reduced further, the compiler did a sufficient job in optimizing the code so that the code will be compute-bound. Performance will not be significantly affected by data movement. This is the case for example for \texttt{gemm} and \texttt{heat-3d}.
\item (6 kernels) The upper bound $\OIup$ stands below the machine balance, meaning that without a fundamental change in the design of the algorithm, it will stay bandwidth-bound. $\OIpluto$ is indeed bandwidth-bound, as any implementation has to be.
This is the case for \texttt{atax} or \texttt{trisolv} for example.
\item (6 kernels) The machine balance stands between our lower ($\OIpluto$) and upper bound ($\OIup$). 
This corresponds to the scenario where our upper bound suggests that there might be room for improvement from a performance point of view. 
This concerns for example \texttt{floyd-warshall} and \texttt{lucdmp}.
A careful look at these two cases shows that they can actually be improved: 
PLuTo did not initially handle \texttt{lucdmp} as well as it should because of the presence of some scalars that did not get expanded into arrays.
Doing this by hand allowed PLuTo to tile it and move the \OI from below the machine balance to above.
Floyd-Warshall is a more involved case: the iteration space decomposition that is discovered by \tool (similar to that of Fig.~\ref{fig:fw-2d})
gives a hint as to how to rewrite the code to make it tileable by PLuTo,
thus leading to an \OI that goes above the machine balance. This is a practical example where \tool helps us discovering more communication-efficient algorithm. (We note that our manual analysis in Section~\ref{ssec:expe1} also
uses this decomposition to derive  $\OIopt$.)
\end{enumerate}

We note that the comparison \OI / \MB is relevant for performance. If we are concerned about energy consumption, then this comparison is not relevant. Looking at Figure 7, any large gaps between 
$\OIpluto$ and $\OIup$ indicate that there might be room for reducing the data movement and thus the energy required for the computation.

\begin{figure*}
	\includegraphics[width=\textwidth]{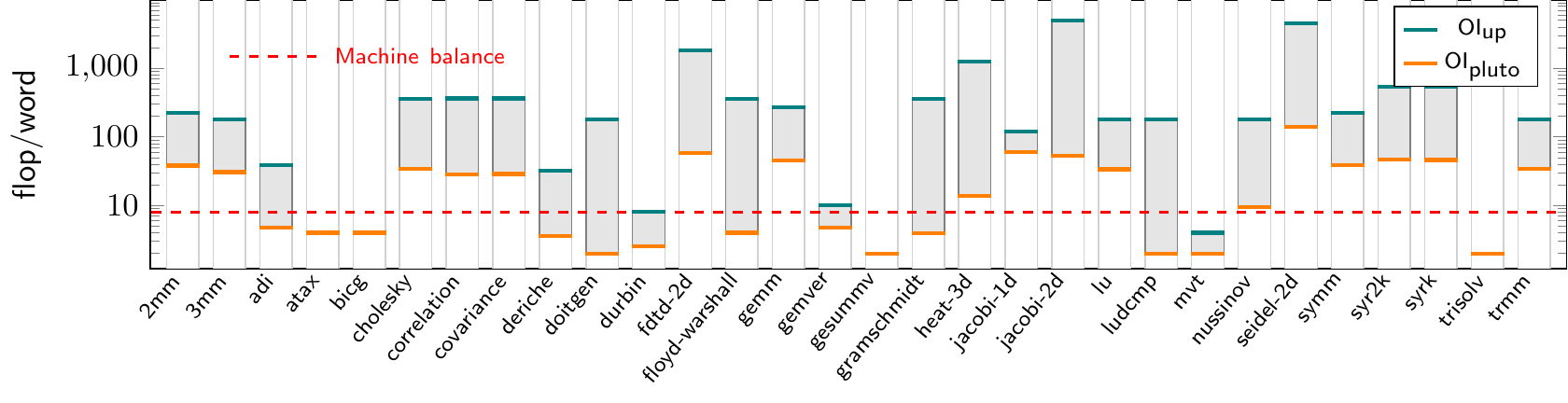}
	\caption{\label{fig:balance}Operational intensity compared to the machine balance for \polybench Benchmarks}
\end{figure*}

\section{Related Work}
\label{sec:related}
The seminal work of Hong \& Kung~\cite{hong.81.stoc} was the first to present
an approach to developing lower bounds on data movement for any valid
schedule of execution of operations in a computational DAG. Their work modeled
data movement in a two-level memory hierarchy and presented manually derived
decomposability factors (asymptotic order complexity, without scaling constants)
for a few algorithms like matrix multiplication and FFT.
Several efforts have sought to build on the fundamental lower bounding approach
devised by Hong \& Kung,
usually targeting one of two objectives:
i) generalizing the cost model to more realistic architecture hierarchies~\cite{savage.cc.95,bilardi2001characterization,bilardi2012lower}, or ii) for providing an \IO complexity with (tight) constant for some specific class of algorithms (sorting/FFT~\cite{aggarwal.ca.88,ranjan11.fft}, relaxation~\cite{ranjan12.rpyr}, or linear algebra~\cite{toledo.jpdc,BDHS11,BDHS11a,DemmelGHL12}).

In the context of linear algebra, Irony et al.~\cite{toledo.jpdc} were the first to use the Loomis-Whitney inequality~\cite{lw49} to find a lower bound on data movement. This was in the context of $\texttt{gemm}$ (one of the kernels of \polybench). The asymptotic upper bound on \OI from this paper is 
$4\sqrt{2}\sqrt{S}$. \tool returns $\sqrt{S}$. This result was then extended in~\cite{acta2014} to 6 more kernels of \polybench:
$\texttt{cholesky}$,
$\texttt{floyd-warshall}$,
$\texttt{lu}$,
$\texttt{symm}$,
$\texttt{syrk}$, and
$\texttt{trmm}$, where their upper bounds on \OI is $8\sqrt{S}$ for all of these kernels.
\tool returns $\sqrt{S}$ for 4 of these kernels, and $2 \sqrt{S}$ for the other 2. The method presented in~\cite{acta2014} is limited to a few algorithms. 
See discussion on ~\cite{Demmel2013TR} for more details on these limitations.

The studies that are the most related to this paper are those from Christ et al.~\cite{Demmel2013TR},
and Elango et al.~\cite{elango-spaa2014,elango-popl2015}. 

The idea of using a variant of the Brascamp-Lieb inequality
to derive bounds for arbitrary affine programs comes from Christ et al.~\cite{Demmel2013TR}. 
However, the approach they propose suffers from several limitations:
1.~The model is based on association of operations with data elements and does not capture data dependencies in a computational DAG. Consequently, it can lead to very weak lower bounds on data movement for computations such as Jacobi stencils.
2.~There is no way to (de-)compose the CDAG, and they view all the statements of the loop body (that has to be perfectly nested) as an atomic statement. As a consequence, it is incorrect to use this approach for loop computations where loop fission is possible.
3.~The lower bounds modeling is restricted to 2S-partitioning, leading to very weak lower bounds for algorithms such as \texttt{adi} or \texttt{durbin}.
4.~Obtaining scaling constants, in particular with non-orthogonal reuse directions, is difficult, and only asymptotic order complexity bounds can be derived.
5.~No automation of the lower bounding approach was proposed, but manually worked out examples of asymptotic complexity as a function of fast memory capacity (without scaling constants) were presented.

Elango et al.~\cite{elango-spaa2014} used a variant of the red-blue pebble game without recomputation, enabling the composition of several sub-CDAG, and the use of a lower-bounding approach based on wavefronts in the DAG.
Manual application of the approach for parallel execution was demonstrated on specific examples, but no approach to automation was proposed.

The later work of Elango et al.~\cite{elango-popl2015} was the first to make the connection between paths in the data-flow graph and regular data reuse patterns and to propose an automated compiler algorithm for affine programs.
However, their proposed approach suffers from several limitations:
1.~Only asymptotic $O(\ldots)$ data movement bounds were obtainable, without any scaling constants.
In contrast, \tool generates meaningful non-asymptotic parametric \IO lower bound formulae. From these formulae, we can derive asymptotic lower bounds with scaling constants, critical for use in deducing upper limits on \OI for a roofline model.
2.~Since they were only trying to provide asymptotic bounds without constants, they did not address (de-)composition (asymptotic bounds can be safely summed up even if they interfere).
Also, they only considered enumerative decomposition, and not dimension decomposition through loop parameterization that is necessary to obtain a tight bound for their Matmult-Seidel illustrative example.
They also only considered the simple non-overlapping notion of interference, and did not allow decomposition of the same statement, required in order to obtain a tight bound for computations like \texttt{floyd-warshall}.
3.~Finally, their approach only used the 2S-partitioning paradigm for lower bounds but not the wavefront-based paradigm, thus leading to very weak bounds for benchmarks such as \texttt{adi} or \texttt{durbin}.

\section{Conclusion}
\label{sec:conclusion}
This paper presents the first compile-time analysis tool to automatically compute a non-asymptotic parametric lower bound on the data movement complexity of an affine program. 
For a cache/scratchpad of limited size $\S$, the minimum required data movement in the two-level memory hierarchy is expressed as a function of $\S$ and program parameters. 
As a result, the tool enables, for a representative class of programs that fits in the polyhedral model, the automated derivation of a bound on the best achievable \OI for all possible valid schedule of a given algorithm.
Its effectiveness has been illustrated on a full benchmark suite of affine programs, the \polybench suite, with results matching or improving over the current state of the art for many of them.

Comparing the achievable \OI with the machine balance \MB is of particular importance as it allows to understand the minimal architectural parameters (e.g. cache size, bandwidth, frequency, etc.) required to support the required inherent data movement of a large class of algorithms such as those used in scientific computing or machine learning.
Affine program regions handled by our automated analysis covers a large proportion of codes in those areas.

\begin{acks}
	This work was supported in part by the \grantsponsor{NSF}{U.S. National Science Foundation}{} awards
	 \grantnum{NSF}{1645514}, \grantnum{NSF}{1645599}, \grantnum{NSF}{1750399} and \grantnum{NSF}{1816793}.
\end{acks}


\newpage
\appendix
\section{Full example: Cholesky decomposition}
\label{appendixA}

\tool uses two proof techniques, namely the $K$-partition and the wavefront based proofs that are respectively described in Sec.~\ref{sec:partition} and Sec.~\ref{sec:wf}. In this section, we demonstrate the complete process on a concrete example: the \texttt{cholesky} kernel. In this example, the $K$-partition method is the method of choice. (So we do not use the wavefront method.) The pseudo-code and associated \dfg for \texttt{cholesky} are reported in Fig.~\ref{fig:chol}.

The \dfg contains three statement vertices $\{S_1,S_2,S_3\}$ (the vertex corresponding to input array \texttt{A} and the corresponding dependences are omitted as they do not play a role in the lower bound derivation).
The main loop of Alg.~\ref{alg:mainloop} iterates on those statements and computes some lower bound complexities for each of them.
We here consider statement $S_3$ for which the $K$-partition reasoning is the one that leads to the largest lower bound of the three.

\begin{figure}[h]
  \begin{subfigure}{\textwidth}
\begin{lstlisting}[language=C]
for(k = 0; k < n; k++) {
	A[k][k] = sqrt(A[k][k]);                 //S1
   for(i = k+1; i < n; i++)
      A[i][k] /= A[k][k];                       // S2
   for(i = k+1; i < n; i++)
      for(j = k+1; j <= i; j++)
          A[i][j] -= A[i][k] * A[j][k];		// S3
}
\end{lstlisting}
	\caption{Source code\label{fig:code-chol}}
  \end{subfigure}

  \begin{subfigure}{\textwidth}
\centering
	\includegraphics{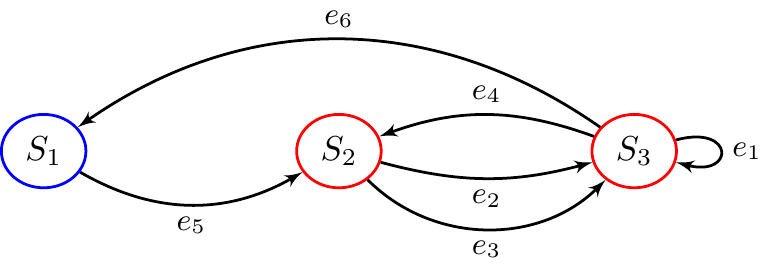}
	\begin{align*}
		R_{e_1} &= \left\{S_3[k-1,i,j] \rightarrow S_3[k,i,j] :\ \ 1\le k < N\ \ \wedge\ \ k+1\le i<N\ \ \wedge\ \ k+1\le j \le i\right\} \\
		R_{e_2} &= \left\{S_2[k,j]\rightarrow S_3[k,i,j] :\ \ 0\le k < N\ \ \wedge\ \ k+1\le i<N\ \ \wedge\ \ k+1\le j \le i\right\} \\
		R_{e_3} &= \left\{S_2[k,i]\rightarrow S_3[k,i,j] :\ \ 0\le k < N \ \ \wedge\ \ k+1 \le i<N\ \ \wedge\ \ k+1 \le j \le i\right\} \\
		R_{e_4} &= \left\{ S_3[k-1,i,k]\rightarrow S_2[k,i] :\ \ 1\le k < N \ \ \wedge\ \ k+1 \le i<N\right\} \\
		R_{e_5} &= \left\{S_1[k]\rightarrow S_2[k,i] :\ \ 0\le k < N \ \ \wedge\ \ k+1 \le i<N\right\} \\
		R_{e_6} &= \left\{ S_1[k-1, k, k]\rightarrow S_1[k]:\ \ 1\le k < N \ \ \wedge\ \ k+1 \le i<N\right\} \\
	\end{align*}
	\caption{\label{fig:dfg-chol}\dfg (input nodes are omitted)}
  \end{subfigure}
  \caption{Cholesky decomposition\label{fig:chol}}
\end{figure}

Out of the six paths, procedure \genpaths will select three ``interesting paths'' for statement $S_3$. These are the three paths pointing to $S_3$, namely:

\begin{align*}
	P_1 &= (e_1) & \textrm{is a chain path}\\
	P_2 &= (e_2) & \textrm{is a broadcast path}\\
	P_3 &= (e_3) & \textrm{is a broadcast path}\\
\end{align*}

The ``$S$-centric'' sub-CDAG is obtained by intersecting the domain of $S$ with the corresponding individual domains of the paths of interests $\P=\{P_1,P_2,P_3\}$ which are:
 \[\domain{P_1} = \left\{S_3[k,i,j]:\ \ 1\le k < N\ \ \wedge\ \ k+1\le i<N\ \ \wedge\ \ k+1\le j \le i\right\}\]
 \[\domain{P_2} = \left\{S_3[k,i,j]:\ \ 0\le k < N\ \ \wedge\ \ k+1\le i<N\ \ \wedge\ \ k+1\le j \le i\right\}\] 
 \[\domain{P_3} = \left\{S_3[k,i,j]:\ \ 0\le k < N\ \ \wedge\ \ k+1\le i<N\ \ \wedge\ \ k+1\le j \le i\right\}\]

Leading to an intersection  domain:
\begin{align*}D_S &:= \domain{P_1} \cap \domain{P_2} \cap \domain{P_3} \\
	&= \left\{S_3[k,i,j]:\ \ 1\le k < N\ \ \wedge\ \ k+1 \le i<N\ \ \wedge\ \ k+1 \le j \le i\right\}
\end{align*}

For those paths, the corresponding projections and kernels are:
\begin{align*}
\phi_1(k, i, j)& = \proj{(1, 0, 0)}{k, i, j} =  (0, i, j)& & \textrm{kernel }k_1 = \Ker(\phi_1) = \gengrp{(1, 0, 0)}\\
\phi_2(k, i, j)& = (k, j)&                                 &\textrm{kernel }k_2 = \Ker(\phi_2) = \gengrp{(0, 1, 0)}\\
\phi_3(k, i, j)& = (k, i)&                                 &\textrm{kernel }k_3 = \Ker(\phi_3) = \gengrp{(0, 0, 1)}
\end{align*}

The explicit embedding $\rho$ of the parametrized CDAG into $E$ is trivial in this case. (See Section~\ref{ssec:embedding}.)

In order to apply Theorem~\ref{thm:bl} to $E$ with $\phi_1$, $\phi_2$ and $\phi_3$, we need to find constant $s_1, s_2$ and $s_3$ such that Equation~(\ref{eq:bl1}) is true. Since the projection kernels, $\Ker(\phi_j)$, are linearly independent, in this case, there is no need to compute the generated lattice of subgroups.
Simply testing on each kernel individually for Equation~(\ref{eq:bl1}) is sufficient (see proof in \cite{Demmel2013TR}, Sec.~6.3). This leads to the following conditions on $s_1, s_2$ and $s_3$:
\begin{equation}
	\label{eq:ex-bl-precond}
	\begin{array}{l}
	0 \leq s_1, s_2, s_3 \leq 1\\
	1 \leq s_2 + s_3 \\
	1 \leq s_1 + s_3 \\
	1 \leq s_1 + s_2 
	\end{array}
\end{equation}

We can apply Theorem~\ref{thm:bl} to $E$ with $\phi_1$, $\phi_2$ and $\phi_3$
for any $s_1, s_2$ and $s_3$ satisfying Equation~(\ref{eq:ex-bl-precond})
to get Equation~(\ref{eq:bl2}). This gives that
\begin{equation}
\label{eq:meagain2}
		\card{E} \leq \card{\phi_1(E)}^{s_1}\cdot \card{\phi_2(E)}^{s_2}\cdot\card{\phi_3(E)}^{s_3}.
\end{equation}

Each of the $\card{\phi_i(E)}$ is bounded by $K$, where $K=\ST$, where $S$ is the cache size, and $T$ is the length of a segment. (See Section~\ref{ssec:embedding} and Equation~(\ref{eq:proj-bnd}): 
\begin{equation}
\label{eq:simple2}
\card{\phi_i(E)}\leq K,~i=1,2,3.
\end{equation}
 
Therefore, denoting $\sigma = s_1 + s_2 + s_3$,  we can bound $\card{E}$ with 
\begin{equation}
\label{eq:simple1}
		\card{E} \leq K^\sigma
\end{equation}

We note that we have introduced 4 parameters: $\T$, $s_1$, $s_2$, $s_3$.
We will choose $s_1$, $s_2$, $s_3$ when we minimize the upper bound $E$ for all possible values of $s_1$, $s_2$, $s_3$. We will choose $T$ when we maximize the lower bound in \IO for all possible values of $\T$. For now, we leave these variables as parameters in the reasoning.

Equation~(\ref{eq:simple1}) is a valid upper bound on the cardinality on any $\K$-bounded set $E$ in the parametrized CDAG.
This upper bound enables us to finish the reasoning to find a lower bound on \IO using the $K$-partitioning method. 

However, we are able to obtain a tighter (higher) \IO lower bound if we can find more constraining inequalities on $\card{\phi_i(E)}$ than Equation~(\ref{eq:simple2}). In order to do so, we use the ``sum-the-projections'' trick. (See Section~\ref{sssec:sum-trick}.) 

In order to use the ``sum-the-projections'' trick, we need to check the independence of the projections, so we first intersect the inverse domains of the different paths which are:
\begin{align*}
  R_{P_1}^{-1}(D) &= \left\{S_3[k,i,j]:\ \ 0\le k < N - 1\ \ \wedge\ \ k+1\le i<N - 1\ \ \wedge\ \ k+1\le j \le i\right\}\\
 R_{P_2}^{-1}(D) &= \left\{S_2[k,j]:\ \ 1\le k < N \ \ \wedge\ \ k+1\le j < N\right\}\\
 R_{P_3}^{-1}(D) &= \left\{S_2[k,i]:\ \ 1\le k < N \ \ \wedge\ \ k+1\le i < N\right\}
\end{align*}

Thus getting:
\begin{align*}
  R_{P_1}^{-1}(D) \cap R_{P_2}^{-1}(D) = \emptyset &&\Rightarrow P_1\textrm{ is independent from }P_2 \\
  R_{P_1}^{-1}(D) \cap R_{P_3}^{-1}(D) = \emptyset &&\Rightarrow P_1\textrm{ is independent from }P_3 \\
  R_{P_2}^{-1}(D) \cap R_{P_3}^{-1}(D) \ne \emptyset &&\Rightarrow P_2\textrm{ interferes with }P_3
\end{align*}

We can thus write, for every $\K$-bounded-set ($\K=\ST$) $E$ in the parametrized CDAG:

\begin{align*}
	\card{\phi_1(E)} + \card{\phi_2(E)} &\leq \K \\
	\card{\phi_1(E)} + \card{\phi_3(E)} &\leq \K
\end{align*}

And summing the two:
\begin{equation}
	\card{\phi_1(E)} + \frac{1}{2} \card{\phi_2(E)} + \frac{1}{2}\card{\phi_3(E)} \leq \K .\\
	\label{eq:ex-sum}
\end{equation}

Since Equation~(\ref{eq:ex-sum}) is more constraining on  $\card{\phi_i(E)}$ than Equation~(\ref{eq:simple2}), it enables us to obtain a tighter (higher) \IO lower bound.

In the framework of Lemma~\ref{lemma:optim}, we call $(\beta_1,\beta_2,\beta_3)=(1,1/2,1/2)$, so that Equation~(\ref{eq:ex-sum})
reads 
$$\beta_1\card{\phi_1(E)} + \beta_2\card{\phi_2(E)} + \beta_3\card{\phi_3(E)} \leq \K .$$

We now use Lemma~\ref{lemma:optim} to bound $|E|$ as follows:

\begin{equation}
	\card{E} \leq \K^{\sigma} \prod_{j=1}^m \left(\frac{s_j}{\beta_j \sigma}\right)^{s_j}.
		\label{eq:meagain}
\end{equation}

Equation~(\ref{eq:meagain}) is similar to Equation~(\ref{eq:U}). It is better than Equation~(\ref{eq:simple1}) since it is able to provide smaller upper bounds on the cardinality of $E$.

We now follow Section~\ref{ssec:lb}. The objective is, w.r.t. the constraints given in Equation~(\ref{eq:ex-bl-precond}) on $s_i$, to minimize the right-hand side of the Equation~(\ref{eq:meagain}). This leads to $s_1 = s_2 = s_3 = \frac{1}{2}$ (and $\sigma = \frac{3}{2}$), that is:
\[
	\card{E} \leq 2 \cdot(\K/3)^{3/2}.
\]

Lemma~\ref{lemma:2S} tells us that, if $U$ is an upper bound on the size of a $\ST$-bounded-set in $G$, then:
\[\Q(G) \ge \T \cdot \left\lfloor \frac{\card{V\setminus \msources{V}}}{U}\right\rfloor -\card{\msources{V}}.\]

Here $V = D\cup R_{P_1}^{-1}(D) \cup R_{P_2}^{-1}(D) \cup R_{P_3}^{-1}(D)$, giving:
\begin{align*}
	V\setminus \msources{V} = \left\{S_3[k,i,j]:\ \ 1\le k < N\ \ \wedge\ \ k+1 \le i<N\ \ \wedge\ \ k+1 \le j \le i\right\} \\
	\msources{V} = \left\{S_3[0,i,j]: 1 \le i<N\ \ \wedge\ \ 1 \le j \le i;\ \  S_2[k,i]:\ \ 1\le k < N \ \ \wedge\ \ k+1\le i < N\right\}
\end{align*}
So (omitting lower-order terms) $\card{V\setminus \msources{V}} = \frac{N^3}{6}$ and $\card{\msources{V}} = N^2$.
Taking for $U$ our upper bound on $\card{E}$ provides the following inequality for which the objective is to set a value for $\T$ that maximizes its right hand side:
\[ \Q \ge \T\times \left\lfloor\frac{N^3/6}{2\cdot (\K/3)^{3/2}}\right\rfloor - N^2 \approx \frac{\T}{\ST^{3/2}} \times \frac{N^3/6}{2 \cdot (1/3)^{3/2}}.\]

Setting $\T=2\S$ (so $\K=\S+\T=3\S$) leads to the following lower bound
\[ \Q \geq (2\S) \times \frac{N^3 / 6}{2 \S^{3/2}} = \frac{N^3}{6 \sqrt{\S}} .\]

Lower order terms have been omitted at a few places in this reasoning so this bound is asymptotic. The full expression for the lower bound found by \tool is given in Table~\ref{tab:polybench}.

\section{Full example: LU decomposition}
\label{appendixB}

Similarly to the previous section, this section illustrates the complete process on another concrete example: the LU decomposition which pseudo-code and associated DFG are reported in Fig~\ref{fig:lu}.

\begin{figure}[h]
  \begin{subfigure}{\textwidth}
\begin{lstlisting}[language=C]
for(k = 0; k < n; k++) {
   for(i = k+1; i < n; i++)
      A[i][k] /= A[k][k];			// S1
   for(i = k+1; i < n; i++)
      for(j = k+1; j < n; j++)
          A[i][j] += A[i][k] * A[k][j];		// S2
}
\end{lstlisting}
\caption{Source code\label{fig:code-lu}}
  \end{subfigure}
  
  \begin{subfigure}{\textwidth}
\centering
	\includegraphics{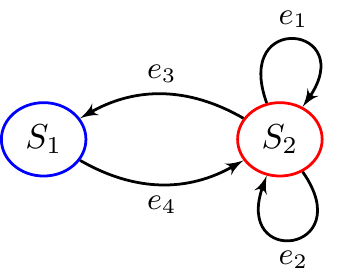}
	\begin{align*}
		R_{e_1} &= \left\{S_2[k-1,i,j] \rightarrow S_2[k,i,j]:\ \ 1\le k < N\ \ \wedge\ \ k+1\le i<N\ \ \wedge\ \ k+1\le j<N\right\} \\
		R_{e_2} &= \left\{S_2[k-1,k,j] \rightarrow S_2[k,i,j]:\ \ 1\le k < N\ \ \wedge\ \ k+1\le i<N\ \ \wedge\ \ k+1\le j<N\right\} \\
		R_{e_3} &= \left\{S_1[k,i] \rightarrow S_2[k,i,j]:\ \ 0\le k < N \ \ \wedge\ \ k+1 \le i<N\ \ \wedge\ \ k+1 \le j<N\right\} \\
		R_{e_4} &= \left\{S_2[k-1,k,k] \rightarrow S_1[k,i]:\ \ 0\le k < N \ \ \wedge\ \ k+1 \le i<N\right\} \\
	\end{align*}
	\caption{\label{fig:dfg-lu}\dfg}
  \end{subfigure}
  \caption{LU decomposition\label{fig:lu}}
\end{figure}

Here, taking $S_2$ as the destination \dfg-node, the following paths are selected:
\begin{itemize}
\item Chain $P_1$: \[R_{P_1} = \left\{S_2[k-1,i,j]\rightarrow S_2[k,i,j]:\ \ 1\le k < N\ \ \wedge\ \ k+1\le i<N\ \ \wedge\ \ k+1\le j<N\right\}\]
    \[\domain{P_1} = \left\{S_2[k,i,j]:\ \ 1\le k < N\ \ \wedge\ \ k+1\le i<N\ \ \wedge\ \ k+1\le j<N\right\}\]

	\item Broadcast $P_2$: \[R_{P_2} = \left\{S_2[k-1,k,j]\rightarrow S_2[k,i,j]:\ \ 1\le k < N\ \ \wedge\ \ k+2\le i<N\ \ \wedge\ \ k+1\le j<N\right\}\]
	  Observe that, as source and image domains have to be disjoint, $S[k, i = k+1, j]$ is excluded from image domain
          \[\domain{P_2} = \left\{S_2[k,i,j]:\ \ 1\le k < N\ \ \wedge\ \ k+2\le i<N\ \ \wedge\ \ k+1\le j<N\right\}\]

	\item Broadcast $P_3$: \[R_{P_3} = \left\{S_1[k,i]\rightarrow S_2[k,i,j]:\ \ 0\le k < N \ \ \wedge\ \ k+1 \le i<N\ \ \wedge\ \ k+1 \le j<N\right\}\]
            \[\domain{P_3} = \left\{S_2[k,i,j]:\ \ 0\le k < N\ \ \wedge\ \ k+1\le i<N\ \ \wedge\ \ k+1\le j<N\right\}\]
\end{itemize}

Intersecting the individual domains with $D_S$ leads to:
\begin{align*}D_S &= \domain{P_1} \cap \domain{P_2} \cap \domain{P_3} \\
	&= \left\{S_2[k,i,j]:\ \ 1\le k < N\ \ \wedge\ \ k+2 \le i<N\ \ \wedge\ \ k+1 \le j<N\right\}
\end{align*}

For those paths, the corresponding projections and kernels are:

\begin{align*}
\phi_1(k, i, j) &= (0, i, j)&& \textrm{kernel }k_1 = \Ker(\phi_1) = \gengrp{(1, 0, 0)}\\
\phi_2(k, i, j) &= (k, k, j)&& \textrm{kernel }k_2 = \Ker(\phi_2) = \gengrp{(0, 1, 0)}\\
\phi_3(k, i, j) &= (k, i)&& \textrm{kernel }k_3 = \Ker(\phi_3) = \gengrp{(0, 0, 1)}
\end{align*}

To check independence of projections, we intersect the inverse domains of the different paths which are:
\begin{align*}
  R_{P_1}^{-1}(D) &= \left\{S_2[k,i,j]:\ \ 0\le k < N-1\ \ \wedge\ \ k+3\le i<N\ \ \wedge\ \ k+2\le j<N\right\}\\
  R_{P_2}^{-1}(D) &= \left\{S_2[k,k+1,j]:\ \ 0\le k < N-1\ \ \wedge\ \ k+2\le j<N\right\}\\
  R_{P_3}^{-1}(D) &= \left\{S_1[k,i]:\ \ 0\le k < N\ \ \wedge\ \ k+2\le i<N\right\}
\end{align*}

They are disjoint so
\begin{equation}\label{eq:lu-tight-sum-trick}
\card{\phi_1(E)} + \card{\phi_2(E)} + \card{\phi_3(E)} \leq 3S.
\end{equation}

The rest is similar to Cholesky, we find
\[
	\card{E} \leq (3S)^{\sum_j s_j} \prod_{j=1}^m \left(\frac{s_j}{\sum_i s_i}\right)^{s_j}.
\]

The constraints on $s_j$'s are the same, so $s_1 = s_2 = s_3 = 1/2$ and
\[
	\card{E} \leq S^{3/2}.
\]

We get
\[ \Q \geq (2S) \frac{N^3 / 3}{S^{3/2}} = \frac{2 N^3}{3 \sqrt{S}}. \]

\section{Complete Lower Bound Formulae for \polybench obtained with the current version of \tool}
\label{appendix:full-bounds}

In Table~\ref{table:complete}, these are the complete formulae as produced by \tool. Next to the complete formulae are 
asymptotic formulae as presented in Table~\ref{tab:polybench}.

We  present the complete formulae produced by \tool for a few reasons. While the lower bounds on \IO obtained by \tool are lower bounds for any values of the parameters ($M$, $N$, $S$, etc), the asymptotic formulae have to be used with care. (1) The asymptotic reasoning, while providing simpler and easier to understand formulae, unfortunately removes the lower bound property. (Negligible negative terms are removed during the asymptotic reasoning.) (2) Also if the asymptotic assumptions are violated, then the asymptotic formulae becomes really off. For example if $S$ is not negligible with respect to $M$ and $N$, or if one dimension in \texttt{gemm} is small. (3) Also, the asymptotic reasonings entail some assumptions. We chose to assume all parameters $N,M\dots$ and cache size $\S$ tend to infinity, and $\S = o(N,M,\dots)$). We can imagine other reasonable asymptotic reasonings. With the complete formula, it is possible to derive them at will. (4) Showing the complete next to the asymptotic expansion explains our asymptotic assumption. (5) It can be instructive to understand the form of the complete formulae returned by \tool.

\newpage
\begin{table}
\scriptsize
\begin{tabular}{|l|l|l|}
\hline
\textbf{kernel}         & 
Complete Lower Bound Formulae for \polybench &
asymptotic simplified
\\
  & 
obtained with the current version of \tool &
 formula
\\
		\hline                                       
\textsf{2mm}            &
$
\max\left(  
N_iN_j+N_jN_k+N_iN_k+N_jN_l+2,  \right.$ &\\&$\left.
\quad\quad\phantom{+}\left( \frac{2}{\sqrt{S}} N_iN_j(N_k-1) + 2N_i + 2N_j + N_k - 4\sqrt{2}S\right)\right.$ & $\phantom{+}\frac{2}{\sqrt{S}} N_iN_jN_k$ \\
&$\left.
\quad\quad         + \left( \frac{2}{\sqrt{S}} N_iN_l(N_j-1) + 2N_i + 2N_l + N_j - 4\sqrt{2}S\right)\right.$ & $         + \frac{2}{\sqrt{S}} N_iN_lN_j$ \\
&$\left.
\quad\quad- 2N_iN_j - 2 
 \right)
$&\\
	 \hline

\textsf{3mm}            &
$
\max\left(  
N_iN_k+N_jN_k+N_jN_m+N_lN_m,\right.$ & \\&$\left.
\quad\quad\phantom{+}\left( \frac{2}{\sqrt{S}} N_iN_j(N_k-1) + 2N_i + 2N_j + N_k - 4\sqrt{2}S\right)\right.$ & $\phantom{+}\frac{2}{\sqrt{S}} N_iN_jN_k$ \\
&$\left.
\quad\quad         + \left( \frac{2}{\sqrt{S}} N_iN_l(N_j-1) + 2N_i + 2N_l + N_j - 4\sqrt{2}S\right)\right.$ & $         + \frac{2}{\sqrt{S}} N_iN_lN_j$ \\
&$\left.
\quad\quad         + \left( \frac{2}{\sqrt{S}} N_jN_l(N_m-1) + 2N_j + 2N_l + N_m - 4\sqrt{2}S\right)\right.$ & $         + \frac{2}{\sqrt{S}} N_jN_lN_m$ \\
&$\left.
\quad\quad- 2N_jN_i 
- 2N_jN_l 
- N_iN_l 
- 6 
 \right)
$ &

\\

	 \hline

\textsf{adi}            &
$ 4N^2 +
\max\left(  
0 , (N^2-4N-S+5)(T-2)
 \right)
$ &$N^2T$\\

	 \hline

\textsf{atax}            &
$ 

MN+N +
\max\left(  
0 ,  \frac18 \frac1S ( (2M-1-8S)(2N-1-8S) -1 ) - 10S + 2 
 \right)
$
& $MN$\\

	 \hline

\textsf{bicg}            &
$ 
MN+M+N
+ \max\left(  
0 ,  
\frac18 \frac1S ( (2M-1-8S)(2N-1-8S) -1 ) - 10S + 2
 \right)
$ & $MN$\\

	 \hline

\textsf{cholesky}            &
$ 
 \max\left(
\frac12  N(N+1),
 \frac16\frac{1}{\sqrt{S}}(N-1)(N-2)(N-3) 
+ \frac1{2\sqrt2}\frac{1}{S}(N-1)(N-2) \right.$ &
$
 \frac16\frac{1}{\sqrt{S}}N^3 
$ \\
&$\left.\quad\quad
- ( N-2)(N-7)
- 4\sqrt2S 
 \right)
$ &
 \\

	 \hline

\textsf{correlation}            &
$ 
 \max\left(
MN+2,
 \frac12 \frac{1}{\sqrt{S}} M(M-1)(N-1+\frac{\sqrt{2}}{2}\frac{1}{\sqrt{S}}) 
 - \frac12 ( M - 3 ) ( M + 2 N - 2 ) 
 + 2 
 - 4S\sqrt{2} \right)
$ &
$
 \frac12 \frac{1}{\sqrt{S}} M^2N
$ \\

	 \hline

\textsf{covariance}            &
$ 
 \max\left(
MN+2,
 \frac12 \frac{1}{\sqrt{S}} M(M-1)(N-1+\frac{\sqrt{2}}{2}\frac{1}{\sqrt{S}}) 
 - \frac12 ( M - 3 ) ( M + 2 N - 2 ) 
 + 1 
 - 4S\sqrt{2} \right)
$ &
$
 \frac12 \frac{1}{\sqrt{S}} M^2N
$ \\
	 \hline

\textsf{deriche}            &
$ 
 HW+1
$ & $ 
 HW
$\\
	 \hline

\textsf{doitgen}            &
$\max\left( 
N_p^2 + N_p N_q N_r,
\frac{2}{\sqrt{S}} N_q N_r N_p (N_p-1+\frac{1}{\sqrt{2}}\frac{1}{\sqrt{S}})
- N_q N_r (N_p-1) 
+ 2 N_p 
- 8\sqrt2 S 
- 1 
 \right)
$ & 
$2\frac{1}{\sqrt{S}} N_q N_r N_p^2$
\\

	 \hline
\textsf{durbin}            &
$ 2N + \max\left( 0,
\frac12 (N-3) (N-2-2S)
 \right)
$ & $\frac12 N^2$\\

	 \hline
\textsf{fdtd-2d}            &
$ \max\left( 
3 N_x N_y - N_y + T - 1 ,
\frac{1}{2\sqrt{2}}\frac{1}{\sqrt{S}} (N_x-2)(N_y-2)(T-1)
 + 2(N_x+2)(N_y+2) 
 \right.$&$ \frac{1}{2\sqrt{2}}\frac{1}{\sqrt{S}} N_xN_yT $\\
 &$\left.\quad\quad
 - T(N_x+N_y-6) 
 - N_y 
 - S 
 - 23 
 \right)
$ &

\\
	 \hline

\textsf{floyd-warshall}            &
$ \max\left( 
N^2,
 \frac{1}{\sqrt{S}}(N-1)^3 - (6N - 19)(N-2) - 8\sqrt2 S 
\right)
$&

$ 
\frac{1}{\sqrt{S}}N^3 
$ 
\\
	 \hline

\textsf{gemm} &
$\max\left( 
N_iN_j+N_jN_k+N_iN_k+2, 
\frac{2}{\sqrt{S}}N_iN_j(N_k-1) + 2N_i + 2N_j + N_k - 4\sqrt{2}S 
 \right)
$ &

$ 
2\frac{1}{\sqrt{S}}N_iN_jN_k 
$

\\

	 \hline

\textsf{gemver} &
$ 
N^2+8N+2 + \max\left(0, 
\frac14 \frac{1}{S}(3N-2)(N-8S)
- 3S 
+ 1 
 \right)
$&$N^2$\\

	 \hline

\textsf{gesummv} &
$ 
2N^2+N+2 + \max\left(0, 
\frac12 \frac{1}{S}(N-1)(N-8S) - 2S 
 \right)
$&$2N^2$\\

	 \hline

\textsf{gramschmidt} &
$ 
 \max\left(MN, 
\frac{1}{\sqrt{S}}MN(N-3)
 - M(N-5-\frac{2}{\sqrt{S}}) 
 - \frac12(N-1)(N-6)
 - 4\sqrt2 S 
 - 3
 \right)
$ &
$\frac{1}{\sqrt{S}}MN^2$

\\

	 \hline

\textsf{heat-3d} &
$ 
 \max\left( 
(N-10)(N+2)^2,
 \frac{9\sqrt[3]{3}}{16}\frac{1}{\sqrt[3]{S}}(T-1)( N-3)^3 
 - 3(T-7)( N-3 )(N-4) \right. $&
 $\frac{9\sqrt[3]{3}}{16}\frac{1}{\sqrt[3]{S}} N^3 T$\\
 &$\left.
 + 42N
 - T
 - \frac{9\sqrt[3]{3}}{4\sqrt[3]{4}} S 
 - 111 
 \right)
$
&

\\

	 \hline

\textsf{jacobi-1d} &
$ 
 \max\left( 
2+n,
  \frac14 \frac{1}{S}(T-1)(N-3) - T - S + 7 
 \right)
$ &
$   \frac14 \frac{1}{S}NT $
\\
	 \hline

\textsf{jacobi-2d} &
$ 
 \max\left(  (N-2)(N+6), 
\frac{2}{3\sqrt3} \frac{1}{\sqrt{S}} (N-3)^2(T-1)
- \frac{4\sqrt2}{3\sqrt3} S
 - (T-7)(2N-7)
 + 14 
 \right)
$ &

$ \frac{2}{3\sqrt3} \frac{1}{\sqrt{S}} N^2T $
\\

	 \hline

\textsf{lu} &

$ 
\max\left( N^2,
\frac23 \frac{1}{\sqrt{S}} (N-2)(N^2-4N+6)
 - 2( N^2 -10N + 18)
 - 8\sqrt{2} S 
 \right)
$

&
$\frac{2}{3}\frac{1}{\sqrt{S}}N^{3}$
\\

	 \hline

\textsf{ludcmp} &
$ 
 \max\left( 
 N^2+N, 
 \frac13 \frac{1}{\sqrt{S}}( 2N -3 )(N-1)(N-2) 
 \sqrt{2} \frac{1}{S}(N-1)(N-2) 
 - ( 2N^2 - 15N + 19 ) 
 - 16\sqrt{2}S 
 \right)
$&
$\frac{2}{3}\frac{1}{\sqrt{S}}N^{3}$
\\

	 \hline

\textsf{mvt} &
$ 
N^2+4N + \max\left( 0,
 \frac16 \frac{1}{S} N(N-1) - 2S - 4N + 4
 \right)
$ & $N^2$\\

	 \hline

\textsf{nussinov} &
$ 
\frac12 N^2 + \frac52 N - 1 + \max\left( 0,
 \frac16 \frac{1}{\sqrt{S}} (N-3)(N-4)(N-5) 
 + \frac14 \frac{1}{S}\sqrt{2}( 3N^2 -19N + 6 ) 
 \right.$&$\frac{1}{6}\frac{1}{\sqrt{S}}N^{3}$\\
 &$\left.
 - ( N^2 -13N + 22 )
 - 8\sqrt{2}S 
 \right)
$ 

&

\\

	 \hline

\textsf{seidel-2d} &
$\max\left( 
N^2,
\frac{2}{3\sqrt{3}}\frac{1}{\sqrt{S}} (N-3)^2 (T-1)
- ( 2N - 7 )(T -5 ) 
- \frac{4\sqrt{2}}{3\sqrt{3}}S
+ 12
 \right)
$ & $ \frac{2}{3\sqrt{3}} \frac{1}{\sqrt{S}} N^2 T $\\
	 \hline

\textsf{symm} &
$\max\left( 
\frac12 M(M+1)+2MN+2, 
   2\frac{1}{\sqrt{S}}(M-1)(M-2) N
 - \frac12 (( 4N + M )  ( M -5  )) 
 \right.$ &
 $2\frac{1}{\sqrt{S}}M^2N$ \\
 &$\left. + 5(M-2) - 8\sqrt{2}S 
 \right)
$ & \\
	 \hline

\textsf{syr2k} &
$\max\left( 
2+2MN+\frac12 N(N+1),
 \frac{1}{\sqrt{S}}(M-1)(N+1)N + M + 4N - 4\sqrt{2} S 
 \right)
$& $\frac{1}{\sqrt{S}}MN^2$\\

	 \hline

\textsf{syrk} &
$\max\left( 
MN+\frac12(N+1)N+2,
\frac12\frac{1}{\sqrt{S}}( M-1)(N+1)N - (M-4)(N-1) - 2\sqrt{2}S + 4
 \right)
$& $\frac{1}{2}\frac{1}{\sqrt{S}}MN^2$\\

\hline

\textsf{trisolv} &
$
\frac12 N (N+1) + N +
\max\left( 0,
\frac{1}{8}\frac{1}{S} ( N-1 )(N-2) 
- 2N
- S 
+ 5 
 \right)
$ & $\frac{1}{2}N^2$\\

	 \hline

\textsf{trmm} &
$\max\left( 
\frac{1}{2}M(M-1)+MN+1,
\frac{1}{\sqrt{S}} ( M - 2 + \frac{\sqrt{2}}{\sqrt{S}} ) ( M  - 1 ) N
- ( M - 4 )( N - 2 )
- 8\sqrt{2}S 
+ 5 

 \right)
$ & $\frac{1}{\sqrt{S}}M^2 N$\\
	 \hline

\end{tabular}
\caption{\label{table:complete}Complete Lower Bound Formulae for \polybench obtained with the current version of \tool}
\end{table}

\end{document}